\documentclass[a4paper,11pt]{article}
\usepackage{a4wide}

\usepackage[utf8]{inputenc}
\usepackage[T1]{fontenc}   
\usepackage[english]{babel} 
\usepackage{amsmath}
\usepackage{amsfonts}
\usepackage{amssymb}
\usepackage{amsthm}
\usepackage{mathrsfs}
\usepackage{calrsfs}
\usepackage[mathscr]{eucal}

\usepackage{url}
\usepackage{enumerate}
\usepackage{graphicx}
\usepackage[all]{xy}
\usepackage{color}
\usepackage{enumerate}
\usepackage{booktabs}
\usepackage{tikz}
\usepackage{algorithm}
\usepackage{caption}
\usepackage{algpseudocode}
\usepackage{algorithm}

\theoremstyle{plain}
\newtheorem{thm}{Theorem}
\newtheorem{theorem}{Theorem}
\newtheorem{prop}[thm]{Proposition}
\newtheorem{proposition}[thm]{Proposition}
\newtheorem{problem}{Problem}

\newtheorem{corollary}[thm]{Corollary}

\newtheorem{lemma}[thm]{Lemma}
\theoremstyle{definition}
\newtheorem{definition}{Definition}

\newtheorem{nota}{Notation}
\newtheorem{ass}{Assumption}
\newtheorem{remark}{Remark}

\theoremstyle{remark}
\newtheorem{rem}{Remark}

{\bfseries}{\itshape}
{\bfseries}{\itshape}

\newcommand{\scp}[2]{{\langle #1,#2\rangle}}
\newcommand{\eqdef}{\mathop{=}\limits^{\triangle}}
\newcommand{\Ft}{\mathbb{F}_2}
\newcommand{\Ac}{{\mathcal A}}
\newcommand{\Cc}{{\mathcal C}}
\newcommand{\Dc}{{\mathcal D}}
\newcommand{\Hc}{{\mathcal H}}

\newcommand{\pbin}{\prob^{\text{bin}}}
\newcommand{\qbin}{q^{\text{bin}}}
\newcommand{\epsbin}{\varepsilon^{\text{bin}}}
\newcommand{\epscon}{\varepsilon^{\text{con}}}
\newcommand{\qcon}{q^{\text{con}}}

\def\sH{{\mathscr H}}

\def\sS{{\mathscr S}}
\def\sT{{\mathscr T}}

\DeclareMathAlphabet{\mathpzc}{OT1}{pzc}{m}{it}

\DeclareMathOperator{\sgn}{sgn}

\newcommand{\prob}{\mathbb{P}}

\newcommand{\pwt}{\pi(\omega,\tau)}
\newcommand{\pwta}[2]{\pi(#1,#2)}
\newcommand{\pwtc}{\pi^{complete}(\omega,\tau)}
\newcommand{\pwtca}[2]{\pi^{complete}(#1,#2)}

\newcommand{\ic}{{\text{\bf i}}}
\newcommand{\Comp}{{\mathbb{C}}}

\begin{document}
\title{Statistical Decoding}

\author{Thomas Debris-Alazard $^*$  \and Jean-Pierre Tillich \footnote{Inria, SECRET Project, 2 Rue Simone Iff 75012 Paris Cedex, France, Email: \texttt{\{thomas.debris,jean-pierre.tillich\}@inria.fr}. Part of this work was supported by the Commission of the European Communities through the Horizon 2020 program under project number 645622 PQCRYPTO.}}
\maketitle

\begin{abstract}
The security of code-based cryptography relies  primarily on the hardness of generic decoding with linear codes. The best generic decoding algorithms are all improvements of an old algorithm due to Prange:  they are known under the name of  information set decoding techniques (ISD).
A while ago a generic decoding algorithm which does not belong to this family was proposed: statistical decoding. 
It is a randomized algorithm that requires the computation of a large set of parity-check equations of moderate weight.
We solve here several open problems related to this decoding algorithm.
 We give in particular the asymptotic complexity of this algorithm, give a rather efficient way of computing the 
parity-check equations needed for it inspired by ISD techniques and give a lower bound on its complexity 
showing that when it comes to decoding on the Gilbert-Varshamov bound it can never be better than Prange's algorithm.
\end{abstract}

\section{Introduction} 	

Code-based cryptography relies crucially on the hardness of decoding generic linear codes.
This problem has been studied for a long time and despite many efforts on this issue \cite{P62,S88,D91,B97b,MMT11,BJMM12,MO15}
the best algorithms for solving this problem \cite{BJMM12,MO15} are exponential in the number of errors that have to be corrected:
correcting $t$ errors in a binary linear code of length $n$  has with the aforementioned algorithms a cost  of
$ 2^{ct(1+ o(1)) }$ where $c$ is a constant depending of the code rate $R$ and the algorithm.
All the efforts that have been spent on this problem have only managed to decrease slightly this exponent $c$.
Let us emphasize that this exponent is  the key for estimating the security level of any code-based cryptosystem.

All the aforementioned algorithms can be viewed as a refinement of the original Prange algorithm \cite{P62} and are actually all 
referred to as ISD algorithms. There is however an algorithm that does not rely at all on Prange's  idea and does not belong to the ISD family: statistical decoding proposed first by Al Jabri in \cite{J01} and improved a little bit by Overbeck in \cite{O06}. Later on,
\cite{FKI07} proposed an iterative version of this algorithm. It is essentially a two-stage algorithm, the first step 
consisting in computing an exponentially large number of parity-check equations of the
smallest possible weight $w$, and then from these parity-check equations the error is recovered by some 
kind of majority voting based on these parity-check equations.

However, even if the study made by R. Overbeck in \cite{O06} lead to the
conclusion that this algorithm did not allow better attacks on the cryptosystems he considered, he did not propose 
 an asymptotic formula of its complexity that would have allowed to conduct a systematic study of the performances
of this algorithm. Such an asymptotic formula has been proposed in \cite{FKI07} through a simplified analysis of
statistical decoding, but as we will see this analysis does not capture accurately 
 the complexity of statistical decoding. Moreover both papers did not assess in general  the complexity of  the first step of the algorithm which consists in computing a large set of parity-check equations of moderate weight.

The primary purpose of this paper is to clarify this matter by giving three results. 
First, we give a rigorous asymptotic study of the exponent $c$ of statistical decoding 
by relying on asymptotic formulas for Krawtchouk polynomials \cite{IS98}. The number of equations which are needed for this method turns out to be
remarkably simple for a large set of parameters. In   Theorem \ref{biasSDecoding} we prove that the number of parity check equations 
of weight $\omega n$ that are needed in a code of length $n$ to decode $\tau n$ errors is of order $O(2^{n(H(\omega)+H(\tau)-1)})$ (when we ignore polynomial factors) and this as soon as $\omega \geq \frac{1}{2} - \sqrt{\tau-\tau^2}$.
For instance, when we consider the hardest instances of the decoding problem which correspond to the 
case where the number of errors is equal to the Gilbert-Varshamov bound, then essentially our results indicate that we have to take 
{\em all} possible parity-checks of a given weight (when the code is assumed to be random)  to perform statistical decoding.
This asymptotic study also allows to
conclude that  the modeling of iterative
statistical decoding made in \cite{FKI07} is too optimistic. Second, inspired by ISD techniques, we propose a rather efficient method for 
computing a huge set of parity-check equations of rather low weight. 
Finally, we give a lower bound on the complexity of this algorithm that shows that it can  not improve upon Prange's algorithm
for the hardest instances of decoding.
 
This lower bound follows by observing that the number $P_w$ of the parity-check equations of weight $w$ that are needed 
for the second step of the algorithm  
is clearly a lower-bound on the complexity of statistical decoding. What we actually prove in the last part 
of the paper is that irrelevant of the way we obtain these parity-check equations in the first step, the lower bound
on the complexity of statistical decoding coming from the infimum of these $P_w$'s is always larger than the complexity of the Prange 
algorithm for the hardest instances of decoding.

\section{Notation}	
	
	As our study will be asymptotic, we neglect polynomial factors and use the following notation:
	\begin{nota}
		Let $f,g : \mathbb{N} \rightarrow \mathbb{R}$, we write $f =  \tilde{O}(g)$ iff there exists 
		a polynomial $P$  such that $f=O(Pg)$.
		\end{nota}

	Moreover, we will often use the classical result $\binom{n}{w} = \tilde{O}\left( 2^{nH\left( \frac{w}{n} \right)} \right)$
	where $H$ denotes the binary entropy. We will also have to deal with complex numbers and follow the convention of the article \cite{IS98} we use here: 
$\ic$ is the imaginary unit satisfying the equation $\ic^2=-1$, $\Re(z)$ is the real part of the complex number $z$ and we choose the branch of the 
complex logarithm with
$$
\ln(z) = \ln|z| + \ic \arg(z), \;\;z \in \Comp \setminus [-\infty,0],
$$
and $\arg(z) \in[-\pi,\pi)$.

		\section{Statistical Decoding}
	\label{sdeco} 

In  the whole paper we consider the computational decoding problem which we define as follows:
	\begin{problem}
		Given a binary  linear code of length $n$ of rate $R$, a word $y \in \mathbb{F}_{2}^{n}$ 
		at distance $t$ from the code, find a codeword $x$ such that $d_{H}(x,y)=t$ where $d_{H}$ denotes the Hamming distance. 
	\end{problem}
	Generally we will specify the code by an arbitrary generator matrix $G$ and we will denote by CSD$(G,t,y)$ a specific instance of this problem.
	We will be interested as is standard in cryptography in the case where 
 $G \in \mathbb{F}_{2}^{Rn\times n}$  \textit{is supposed to be random}.
 
The idea behind statistical decoding may be described as follows.
We first compute a very large set $\sS$ of parity-check  equations of some weight $w$ and compute all
scalar products $\langle y,h \rangle$ (scalar product is modulo $2$)
for $h \in \sS$. It turns out that if we consider only the parity-checks involving a given code position $i$ the
scalar products have a probability of being equal to $1$ which depends whether there is an error in this position or not.
Therefore counting the number of times when $\langle y,h \rangle=1$ allows to recover the error in this position.

Let us analyze now this algorithm more precisely.
To make this analysis tractable we will need to make a few simplifying assumptions.
The first one we make is the same as the one made by R. Overbeck in \cite{O06}, namely that 
\begin{ass}\label{ass:one}
The distribution of the $\langle y,h \rangle$'s when $h$ is drawn uniformly at random from the dual codewords of weight $w$ is approximated by the distribution of  $\langle y,h \rangle$ when
$h$ is drawn uniformly at random among the words of weight $w$.
\end{ass}

A much simpler model is given in \cite{FKI07} and is based on modeling the distribution of the $\scp{y}{h}$'s
as the distribution of $\scp{y}{h}$ where the coordinates of $h$ are i.i.d. and distributed as a Bernoulli variable of parameter $w/n$.
This presents the advantage of making the analysis of statistical decoding much simpler and allows to analyze more refined versions of statistical decoding.
However as we will show, this is an oversimplification and leads to an over-optimistic estimation of the complexity of 
statistical decoding.
 The following notation will be useful.

\begin{nota}{ }\mbox{} \\
$\cdot$ $S_w \eqdef \{x \in \Ft^n : w_H(x)=w\}$ denotes the set of binary of words of length $n$ of weight $w$; \\
$\cdot$ $S_{w,i} \eqdef \{x \in S_w: x_i = 1\}$;\\
 $\cdot$ $\sH_w \eqdef \Cc^\perp \cap S_w$;\\
$\cdot$ $\sH_{w,i} \eqdef \Cc^\perp \cap S_{w,i}$;\\
 $\cdot$ $X \sim \mathcal{B}(p)$ means that $X$ follows a Bernoulli law of parameter $p$ ;\\
$\cdot$ $h \sim S_{w,i}$ means we pick $h$ uniformly at random in  $S_{w,i}$.
\end{nota}

\subsection{Bias in the parity-check sum distribution}
	\label{bias}

We start the analysis of statistical decoding by computing  the following probabilities which approximate the
true probabilities we are interested in (which correspond to choosing $h$ uniformly at random
in $\sH_{w,i}$ and not in $S_{w,i}$) under Assumption \ref{ass:one}
	\[ q_{1}(e,w,i) = \mathbb{P}_{h \sim S_{w,i}} \left( \langle e,h \rangle = 1 \right) \mbox{ when } e_{i} = 1  \]
	\[ q_{0}(e,w,i) = \mathbb{P}_{h \sim S_{w,i}} \left( \langle e,h \rangle = 1 \right) \mbox{ when } e_{i} = 0 \]
 These probabilities are readily seen to be equal to
	\[ q_{1}(e,w,i) = \frac{\mathop{\sum}\limits_{j \mbox{ \tiny{even}}}^{w-1} \binom{t-1}{j} \binom{n-t}{w-1-j}} {\binom{n-1}{w-1}} \]
	\[ q_{0}(e,w,i) = \frac{\mathop{\sum}\limits_{j \mbox{ \tiny{odd}}}^{w-1} \binom{t}{j}  \binom{n-t-1}{w-1-j}} {\binom{n-1}{w-1}}  \]
They are independent of the error and the position $i$. So, in the following we will use the notation $q_{1}$ and $q_{0}$.
We will define the biases $\varepsilon_{0}$ and $\varepsilon_{1}$ of statistical decoding by
	\[ q_{0} = \frac{1}{2} + \varepsilon_{0} \mbox{ } ; \mbox{ }  q_{1} = \frac{1}{2} + \varepsilon_{1}   \]
It will turn out, and this is essential, that $\varepsilon_0 \neq \varepsilon_1$.	
 We can use these biases ``as a distinguisher''. They are at the heart of statistical decoding. 
Statistical decoding is nothing but a statistical hypothesis testing algorithm distinguishing between two hypotheses :
	\[ \Hc_0 \mbox{ } : \mbox{ } e_{i} = 0 \quad ; \quad \Hc_{1} \mbox{ } : \mbox{ } e_{1} = 1 \] based on computing 
the random variable $V_{m}$ for $m$ uniform and independent draws of vectors in $\sH_{w,i}$:
	\[ V_{m} = \sum_{k=1}^{m} \sgn(\varepsilon_{1} - \varepsilon_{0}) \cdot \langle y,h^{k} \rangle \in \mathbb{Z}  \]

	We have $\langle y,h^{k}\rangle \sim \mathcal{B}(1/2 + \varepsilon_{l})$ according to $\mathcal{H}_{l}$. So the expectation of $V_{m}$ is given under $\mathcal{H}_{l}$ by:
	\[ E_{l} = m \sgn(\varepsilon_{1} - \varepsilon_{0}) (1/2 + \varepsilon_{l}) \]

	We point out that we have  $E_{1} > E_{0}$ regardless of the term $\sgn(\varepsilon_{1} - \varepsilon_{0})$. In order to apply the following proposition, we make the following assumption:
	\begin{ass}
		\label{ass:two}
		$\langle y,h^{k} \rangle$ are independent variables. 
	\end{ass}

	\begin{proposition}
		[Chernoff's Bound]$ $ Let $0 < p < 1$, $Y_{1},\cdots,Y_{m}$ i.i.d $\sim \mathcal{B}(p)$ and we set $Z_{m} = \sum_{k=1}^{m} Y_{k}$. Then,
		\[ \forall t \geq 0, \quad \mathbb{P}\left( | Z_{m} - mp | \geq m \delta \right) \leq 2e^{-2m\delta^{2}} \]
	\end{proposition}
	
%
%
%
%
%
	\textbf{Consequences:}
	Under $\mathcal{H}_{l}$, we have
	\[ \mathbb{P}\left( | V_{m} - m\sgn(\varepsilon_{1}- \varepsilon_{0})\cdot (1/2 + \varepsilon_{l}) | \geq m \cdot \frac{|\varepsilon_{1} - \varepsilon_{0}|}{2} \right) \leq 2\cdot 2^{-m \cdot \frac{(\varepsilon_{1} - \varepsilon_{0})^{2}}{2\ln(2)}} \]

	To take our decision we proceed as follows: if $V_{m} < \frac{E_{0} + E_{1}}{2}$ where
	\[ \frac{E_{1} + E_{0}}{2} = \frac{m}{2} \sgn(\varepsilon_{1} - \varepsilon_{0}) (1+\varepsilon_{1} + \varepsilon_{0}) \] 
	we choose $\mathcal{H}_{0}$ and $\mathcal{H}_{1}$ if not. 
For the cases of interest to us (namely $w$ and $t$ linear in $n$)  the bias $\varepsilon_1 - \varepsilon_0$ is an exponentially small function of the 
codelength $n$ and it is obviously enough to choose $m$ to be of order $O\left( \frac{\log n}{(\varepsilon_{1}-\varepsilon_{0})^{2}}\right)$ to be able
to make the good decisions on all $n$ positions simultaneously.

\par{\em On the optimality of the decision.}
All the arguments used for distinguishing both hypotheses are very crude and this raises the question whether a better test exists. It turns out that in the regime of interest to us, namely $t$ and $w$ linear in $n$, the term $\tilde{O}\left( \frac{1}{(\varepsilon_{1}-\varepsilon_{0})^{2}}\right)$ is of the right order. Indeed our statistical test amounts actually to the Neymann-Pearson test (with a threshold in this case which is not necessarily in the middle, i.e. equal to $m\frac{1+\varepsilon_0+\varepsilon_1}{2}$). In the case of interest to us, the bias between both distributions $\varepsilon_1 - \varepsilon_0$ is exponentially small in $n$ and Chernoff's bound captures accurately the large deviations of the random variable $V_m$. Now we 
could wonder whether using some finer knowledge about the hypotheses $\Hc_0$ and $\Hc_1$ could do better. For instance we know the a priori probabilities 
of these hypotheses since $\prob(e_i=1)=\frac{t}{n}$. It can be readily verified that using Bayesian hypothesis testing based on the a priori knowledge of the a priori probabilities of both hypotheses does not allow to change
the order of number of tests which is still $\tilde{O}\left( \frac{1}{(\varepsilon_{1}-\varepsilon_{0})^{2}}\right)$ when $t$ and $w$ are linear in $n$.

	\subsection{The statistical decoding algorithm} 
	\label{sdecobias} 
	
	Statistical decoding is a randomized algorithm which uses the previous distinguisher. As we just noted, this distinguisher needs $\tilde{O}\left( \frac{1}{(\varepsilon_{1}-\varepsilon_{0})^{2}}\right)$ parity-check equations of weight $w$ to work. This number obviously depends on $w,R$ and $t$ and we use the notation:

	\begin{nota}
		$P_{w} \mathop{=}\limits^{\triangle} \frac{1}{(\varepsilon_{1} - \varepsilon_{0})^{2}}$. 
	\end{nota}

	Now we have two frameworks to present statistical decoding. We can consider the computation of $\tilde{O}(P_{w})$ parity-check equations as a pre-computation or to consider it as a part of the algorithm. To consider the case of pre-computation, simply remove Line $4$ of  Algorithm 1   
	and consider the $\sS_{i}$'s as an additional input to the algorithm. \texttt{ParityCheckComputation}$_{w}$ will denote an algorithm which for an input $G,i$ outputs $\tilde{O}(P_{w})$ vectors of $\sH_{w,i}$.

		\begin{algorithm}
			\caption{\texttt{DecoStat} : \textbf{Statistical Decoding}}
		\begin{algorithmic}[1]
				\State $Input : G \in \mathbb{F}_{2}^{Rn \times n}, y = xG + e \in \mathbb{F}_{2}^{n}, w \in \mathbb{N}$ 
				\State $Output : e$  /*\textit{Error Vector}*/
				\For{$i = 1 \cdots n$}
				\State $\sS_{i} \leftarrow \texttt{ParityCheckComputation}_{w}(G,i)$ /*\textit{Auxiliary Algorithm}*/
				\State $V_{i} \leftarrow 0$
				\ForAll{$h \in \sS_{i}$}
				\State $V_{i} \leftarrow V_{i} + \sgn(\varepsilon_{1} - \varepsilon_{0})\cdot \langle y,h  \rangle $
				\EndFor 
				\If{$V_{i} < \sgn(\varepsilon_{1} - \varepsilon_{0}) P_{w}  \frac{1+\varepsilon_{1}+\varepsilon_{0}}{2}$}
				\State $e_{i} \leftarrow 0$
				\Else 
				\State $e_{i} \leftarrow 1$
				\EndIf
				\EndFor 
				\State \Return $e$
			\end{algorithmic}
		\end{algorithm}
		
Clearly statistical decoding complexity is given by	
\begin{itemize}
			\item When the $\sS_{i}$'s are already stored and computed: $\tilde{O} \left( P_{w} \right)$;
			
			\item When the  $\sS_{i}$'s have to be computed: $\tilde{O}\Big( P_{w} + |\emph{\texttt{PC}$_{w}$}| \Big) $
where $|\emph{\texttt{PC}$_{w}$}|$ stands for the complexity of the call \texttt{ParityCheckComputation}$_{w}$.
		\end{itemize}	

	As explained in introduction, our goal is to give the asymptotic complexity of statistical decoding. We introduce for this purpose the following notations:
	\begin{nota}$ $
		
		$\cdot$ $\omega \eqdef \frac{w}{n}$;

		$\cdot$ $\tau \eqdef \frac{t}{n}$.

	\end{nota}
	The two following quantities will be the central object of our study.
	\begin{definition}[Asymptotic complexity of statistical decoding]
	We define the asymptotic complexity of statistical decoding when the $\sS_i$'s are already computed by 
	$$
	\pwt \eqdef \varliminf_{n \to + \infty} \frac{1}{n} \log_{2} P_{w}
	$$
	whereas the asymptotic complexity of the complete algorithm of statistical decoding (including the computation of
	the parity-check equations) is defined by 
	$$
	\pwtc \eqdef \varliminf_{n \to + \infty}  \frac{1}{n} \max 
	\Big( \log_{2} P_{w}, \log_{2}|\text{\upshape{\texttt{ParityCheckComputation}}}_{w}| \Big).
	$$
	\end{definition}
	\begin{rem}
	One could wonder why these quantities are defined as infimum limits and not directly as limits. This is due to the fact that in certain regions
	of the error weight and parity-check weights the asymptotic bias may from time to time become much smaller than it typically is. This bias is indeed proportional
	to values taken by a Krawtchouk polynomial and for certain errors weights and parity-check weights we may be close to the zero of the 
	relevant Krawtchouk polynomial (this corresponds to the second case of Theorem \ref{th:expansion}).
	\end{rem}

	 We are looking for explicit formulas for $\pwt$ and $\pwtc$. The second quantity depends 
on  the algorithm which is used. We will come back to this issue in 
Subsection \ref{fram}.
For our purpose we will use Krawtchouk polynomials and asymptotic expansions for them coming from \cite{IS98}.
Let $m$ be a positive integer, we recall that the Krawtchouk polynomial of degree $v$ and order $m$, $p_v^m(X)$ is  defined for $v  \in \{0,\cdots,m\}$ by:
	\[ p_{v}^{m}(X) = \frac{(-1)^{v}}{2^{v}}  \sum_{j=0}^{v} (-1)^{j} \binom{X}{j}  \binom{m-X}{v-j} \quad \mbox{where, } \binom{X}{j} = \frac{1}{j!}\left( X  (X-1) \cdots (X-j+1) \right) \]

These Krawtchouk polynomials are readily related to our biases.
We can namely observe that $\sum_{j=0}^{w-1} \binom{t-1}{j}  \binom{n-t}{w-1-j} = \binom{n-1}{w-1}$ to recast the following 
evaluation of a Krawtchouk polynomial as
\begin{eqnarray}
- \frac{(-2)^{w-2}}{\binom{n-1}{w-1}}  p_{w-1}^{n-1}(t-1) &= & \frac{\sum_{j=0}^{w-1} (-1)^{j} \binom{t-1}{j}  \binom{n-t}{w-1-j}}{2 \binom{n-1}{w-1}}
\nonumber\\
& = & \frac{\sum_{\substack{j=0 \\ j \text{ even}}}^{w-1}  \binom{t-1}{j}  \binom{n-t}{w-1-j}
- \sum_{\substack{j=1 \\ j \text{ odd}}}^{w-1}  \binom{t-1}{j}  \binom{n-t}{w-1-j} }{2 \binom{n-1}{w-1}} \nonumber\\
& = & \frac{2 \sum_{ \substack{j=0 \\ j \text{ even}}}^{w-1}  \binom{t-1}{j}  \binom{n-t}{w-1-j}
- \binom{n-1}{w-1} }{2\binom{n-1}{w-1}}\nonumber\\
& = & \varepsilon_1  \label{eq:epsilon1}
\end{eqnarray}
We have a similar computation for $\varepsilon_0$ 
\begin{eqnarray}
 \frac{(-2)^{w-2}}{\binom{n-1}{w-1}}  p_{w-1}^{n-1}(t) &= & - \frac{\sum_{j=0}^{w-1} (-1)^{j} \binom{t}{j}  \binom{n-1-t}{w-1-j}}{2 \binom{n-1}{w-1}}
 \nonumber\\
& = &- \frac{\sum_{\substack{j=0 \\ j \text{ even}}}^{w-1}  \binom{t}{j}  \binom{n-1-t}{w-1-j}
- \sum_{\substack{j=1 \\ j \text{ odd}}}^{w-1}  \binom{t}{j}  \binom{n-1-t}{w-1-j} }{2 \binom{n-1}{w-1}} \nonumber \\
& = &- \frac{\binom{n-1}{w-1}- 2 \sum_{ \substack{j=0 \\ j \text{ odd}}}^{w-1}  \binom{t}{j}  \binom{n-1-t}{w-1-j}
 }{2\binom{n-1}{w-1}} \nonumber\\
& = & \varepsilon_0 \label{eq:epsilon0}
\end{eqnarray}

Let us recall  Theorem 3.1 in \cite{IS98}.
\begin{theorem}[{\cite[Th. 3.1]{IS98}}]\label{th:expansion}
Let $m,v$ and $s$ be three positive integers.
We set $\nu \eqdef \frac{v}{m}, \alpha \eqdef \frac{1}{\nu}$ and $\sigma = \frac{s}{m }$.
We assume $\alpha \geq 2$.
Let $$p(z)  =   \log_2 z - \frac{\sigma}{\nu} \log(1+z) - \left(\alpha - \frac{\sigma}{\nu} \right) \log_2(1-z).$$
$p'(z)=0$ has two solutions $x_1$ and $x_2$ which are the two roots of the equation
$
(\alpha - 1) X^{2} + (\alpha - 2\frac{\sigma}{\nu}) X + 1 =0
$.
Let $D \eqdef \left(\alpha- 2 \frac{\sigma}{\nu} \right)^2-4(\alpha-1)$ and 
$\Delta \eqdef \alpha - \frac{2\sigma}{\nu}$. The two roots are equal to $\frac{- \Delta \pm  \sqrt{D}}{2 (\alpha-1)}$ and $x_1$ is defined to be 
root $ \frac{- \Delta + \sqrt{D}}{2 (\alpha-1)}$.
There are two cases to consider
\begin{itemize}
\item
In the case $\frac{\sigma}{\nu} \in ( 0,\alpha/2 - \sqrt{\alpha- 1})$, $D$ is positive, $x_1$ is a real negative number and we can write
\begin{equation}
\label{eq:real}
p_{v}^{m}(s) = Q_{\sigma,\nu}(v) 2^{-(p(x_1)+1) v}
\end{equation}
where $Q_{\sigma,\nu}(v)  \eqdef  -\sqrt{\frac{1-r^2}{2 \pi r D v}}( 1 + O(v^{-1/2}))$
and $r\eqdef -x_1$.
\item In the case $\frac{\sigma}{\nu} \in ( \alpha/2 - \sqrt{\alpha- 1},\alpha/2)$, $D$ is negative, $x_1$ is a complex number
and we have
\begin{equation}
\label{eq:complex}
p_{v}^{m}(s) = R_{\sigma,\nu}(v) \Im\left( \frac{2^{-(p(x_1)+1) v}}{x_1 \sqrt{2p"(x_1)}} (1+ \delta(v))  \right)
\end{equation}
where $\Im(z)$ denotes the imaginary part of the complex number $z$,
$\delta(v)$ denotes a function  which is $o(1)$ uniformly in $v$,
and
$
R_{\sigma,\nu}(v) \eqdef  \frac{1+O(v^{-1/2})}{\sqrt{\pi v}}
$.
\end{itemize}
The asymptotic formulas hold uniformly on the compact subsets of the corresponding open intervals.
\end{theorem}

\begin{remark}
Note that strictly speaking \eqref{eq:real} is incorrectly stated in \cite[Th. 3.1]{IS98}. The problem is that (3.20) is incorrect 
in \cite{IS98}, since both $p"(-r_1)$ and $p^{(3)}(-r_1)$ are negative and taking a square root of these expressions leads to a purely 
imaginary number in (3.20). This can be easily fixed since the expression which is just above (3.20) is correct and it just remains
to take the imaginary part correctly to derive \eqref{eq:real}.
\end{remark}

It will be helpful to use the following notation from now on.
\begin{nota}
\begin{eqnarray*}
m & \eqdef & n-1 \\
v & \eqdef & w-1\\
\nu & \eqdef & \frac{v}{m}\\
\alpha & \eqdef & \frac{1}{\nu}\\
\sigma_0 & \eqdef & \frac{t}{m}\\
\sigma_1 & \eqdef & \frac{t-1}{m}
\end{eqnarray*}
and for $i \in \{0,1\}$ we define the following quantities
\begin{eqnarray*}
p_i(z) & \eqdef &   \log_2 z - \frac{\sigma_i}{\nu} \log(1+z) - \left(\alpha - \frac{\sigma_i}{\nu} \right) \log_2(1-z) \\
\Delta_i &\eqdef &\alpha - \frac{2\sigma_i}{\nu} \\
D_i &\eqdef &\left(\alpha- 2 \frac{\sigma_i}{\nu} \right)^2-4(\alpha-1) \\
z_i &\eqdef &\frac{- \Delta_i + \sqrt{D_i}}{2 (\alpha-1)} 
\end{eqnarray*}
\end{nota}

We are now going use these asymptotic expansions to derive explicit formulas for $\pwt$.
We start with the following lemma.

\begin{lemma}\label{lem:real}
With the hypothesis of Proposition just above, we have
 \[  \frac{\varepsilon_{0}}{\varepsilon_{1}} = - \frac{1+z_1}{1-z_1} \left(1  + O(w^{-1/2}\right). \]
\end{lemma}
\begin{proof}
From \eqref{eq:epsilon1} and \eqref{eq:epsilon0} we have
\begin{equation}\label{eq:ratio}
\frac{\varepsilon_0}{\varepsilon_1} = - \frac{p_{w-1}^{n-1}(t)}{p_{w-1}^{n-1}(t-1)}
\end{equation}

By using Theorem \ref{th:expansion} we obtain when plugging the asymptotic expansions of the Krawtchouk polynomials into
\eqref{eq:ratio}
\begin{eqnarray}
\frac{\varepsilon_0}{\varepsilon_1}& =& - \frac{Q_{\sigma_0,\nu}(v) 2^{-p_0(z_0) v}}{Q_{\sigma_1,\nu}(v) 2^{-p_1(z_1) v}} \nonumber\\
& =& - \frac{Q_{\sigma_0,\nu}(v)}{Q_{\sigma_1,\nu}(v) }2^{(p_1(z_1)-p_0(z_0))v} \label{eq:ratio0}
\end{eqnarray}

We clearly have $\sigma_1 = \sigma_0 - \frac{1}{m}$ and $z_1 = z_0 + O\left( \frac{1}{m} \right)$ and
therefore from the particular form of $Q_{\sigma_i,\nu}(v)$  we deduce that
\begin{equation}
 \frac{Q_{\sigma_0,\nu}(v)}{Q_{\sigma_1,\nu}(v) } = 1+ O(v^{-1/2}).
\end{equation}
We observe now that 
\begin{eqnarray}
\frac{\sigma_1}{\nu} - \frac{\sigma_0}{\nu}& = & \frac{t-1}{v} - \frac{t}{v}\\
& = & - \frac{1}{v}
\end{eqnarray}
and therefore 
\begin{eqnarray}
(p_1(z_1) - p_0(z_0))v\!\!\! & = & \!\!\!\left( \log_{2}(z_1) - \frac{\sigma_1}{\nu} \log_{2}(1+z_1) - (\alpha - \frac{\sigma_1}{\nu}) \log_{2}(1-z_1)  \right. \nonumber\\
& & \left. - \log_{2}(z_0) + \frac{\sigma_0}{\nu} \log_{2}(1+z_0) + (\alpha - \frac{\sigma_0}{\nu}) \log_{2}(1-z_0) \right) v \nonumber \\
& = &\!\!\! \left( \log_2 \frac{z_1}{z_0} - \frac{\sigma_0}{\nu} \log_{2}\frac{1+z_1}{1+z_0}  - (\alpha - \frac{\sigma_0}{\nu}) \log_{2}\frac{1-z_1}{1-z_0}
+ \frac{1}{v} \log_2(1+z_1) - \frac{1}{v} \log_2(1-z_1) \right) v \nonumber\\
& = & \!\!\!\left( \log_2 \frac{z_1}{z_0} - \frac{\sigma_0}{\nu} \log_{2}\frac{1+z_1}{1+z_0}  - (\alpha - \frac{\sigma_0}{\nu}) \log_{2}\frac{1-z_1}{1-z_0}
  \right) v + \log_2 \frac{1+z_1}{1-z_1} \label{eq:end}
\end{eqnarray}

It is insightful to express the term 
$ \log_2 \frac{z_1}{z_0} - \frac{\sigma_0}{\nu} \log_{2}\frac{1+z_1}{1+z_0}  - (\alpha - \frac{\sigma_0}{\nu}) \log_{2}\frac{1-z_1}{1-z_0}
  $
as 
\begin{eqnarray*}
\log_2 \frac{z_1}{z_0} - \frac{\sigma_0}{\nu} \log_{2}\frac{1+z_1}{1+z_0}  - (\alpha - \frac{\sigma_0}{\nu}) \log_{2}\frac{1-z_1}{1-z_0}
& =& p_0(z_1) - p_0(z_0)
\end{eqnarray*}
 
The point is that $p_0'(z_0)=0$   and  $z_1 = z_0+ \delta$ where $\delta = O(1/m)$. Therefore 
$$
p_0(z_1)-p_0(z_0)=p_0(z_0 + \delta)-p_0(z_0) = O(\delta^2) = O(1/m^2).
$$
Using this in \eqref{eq:end}  and then in \eqref{eq:ratio0} implies the lemma.
\end{proof}

From this lemma we can deduce that
\begin{lemma}\label{lem:realcomplete}
Assume $\alpha \geq 2$ and $\frac{\sigma_i}{\nu} \in ( 0, \alpha/2 - \sqrt{\alpha -1} )$ for $i \in \{0,1\}$. We have 
 \[  \varepsilon_{0} -\varepsilon_{1} = (-1)^v\sqrt{\frac{(1+z_1)(1-\nu)}{(z_1-z_1^2)D_1}}2^{-\left(p(z_1)+\frac{H(\nu)}{\nu}\right)v}(1 + O(w^{-1/2})). \]
\end{lemma}
\begin{proof}
We have
\begin{eqnarray}
\varepsilon_{0} -\varepsilon_{1} &=& - \varepsilon_1 \frac{1+z_1}{1-z_1}(1 + O(w^{-1/2})) - \varepsilon_1 \nonumber\\
&= & - \varepsilon_1 \left( \frac{1+z_1}{1-z_1}(1 + O(w^{-1/2}) + 1\right) \nonumber\\
&= & - 2 \varepsilon_1 \left( \frac{1}{1-z_1} + O\left(w^{-1/2}\frac{1+z_1}{1-z_1}\right) \right) \nonumber\\
& = & - \frac{2 \varepsilon_1}{1-z_1}(1 + O(w^{-1/2}) \;\;\text{ (since $-1 \leq z_i \leq 0$)} \nonumber\\
& = & - \frac{(-2)^{v-1}\sqrt{2 \pi v  (1-\nu)}}{2^{\frac{vH(\nu)}{\nu}}}Q_{\sigma_1,\nu}(v) 2^{-(p(z_1)+1)v} \frac{2}{1-z_1}(1 + O(w^{-1/2})\label{eq:milieu}\\
& = &  \frac{(-2)^{v}\sqrt{2 \pi v  (1-\nu)}}{2^{\frac{vH(\nu)}{\nu}}}\sqrt{\frac{1-z_1^2}{-2 \pi z_1 D_1 v }} 2^{-(p(z_1)+1)v} \frac{2}{1-z_1}(1 + O(w^{-1/2}) \nonumber \\
& = & (-1)^v\sqrt{\frac{(1+z_1)(1-\nu)}{(z_1-z_1^2)D_1}}2^{-\left(p(z_1)+\frac{H(\nu)}{\nu}\right)v}(1 + O(w^{-1/2})
\end{eqnarray}
where we used in \eqref{eq:milieu} 
\begin{eqnarray*}
\binom{m}{v} & = & \frac{2^{mH(\nu)}}{\sqrt{2 \pi m \nu (1-\nu)}} \\
& = & \frac{2^{\frac{vH(\nu)}{\nu}}}{\sqrt{2 \pi v  (1-\nu)}}
\end{eqnarray*}
\end{proof}
 The second case corresponding to $\frac{\sigma_i}{\omega} \in (  \alpha/2 - \sqrt{\alpha -1},\alpha/2 )$ is handled by the following lemma
(note that it is precisely the ``sin'' term that appears in it that lead us to define $\pwt$ as an infimum limit and not as a limit)
\begin{lemma}\label{lem:complex}
When  $\frac{\sigma_i}{\omega} \in (  \alpha/2 - \sqrt{\alpha -1},\alpha/2 )$ for $i \in \{0,1\}$ we have
 \[  \varepsilon_1 - \varepsilon_0 =     
 \frac{(-1)^v \sqrt{1-\nu}}{\left|(z_0-z_0^2)\sqrt{p_0"(z_0)}\right|} 2^{-v \left( \Re(p_0(z_0))+ \frac{H(\nu)}{\nu}  \right)}
\sin \left( v \theta - \theta_0 +o(1) \right) (1+o(1))\]
 where  
 $\theta \eqdef \arg \left(2^{-p_0(z_0)}\right)$ and $\theta_0 \eqdef \arg \left( (z_0-z_0^2)\sqrt{p_0"(z_0)}  \right)$.
\end{lemma}
\begin{proof}
The proof of this lemma is very similar to the proof of Lemma \ref{lem:real}.
From \eqref{eq:epsilon1} and \eqref{eq:epsilon0} we have
\begin{equation}\label{eq:ratiodiff}
\varepsilon_1 - \varepsilon_0 = - \frac{(-2)^{w-2}}{\binom{n-1}{w-1}} \left(p_{w-1}^{n-1}(t) + p_{w-1}^{n-1}(t-1) \right)
\end{equation}

By plugging the asymptotic expansion of Krawtchouk polynomials 
given in  Theorem \ref{th:expansion} into
\eqref{eq:ratiodiff} we obtain
\begin{eqnarray*}
\varepsilon_1 - \varepsilon_0 & =& - \frac{(-2)^{w-2}}{\binom{n-1}{w-1}} \left( 
R_{\sigma_1,\nu}(v) \Im\left( \frac{2^{-(p_1(z_1)+1) v}}{z_1 \sqrt{2p_1"(z_1)}} (1+ \delta_1(v))  \right)
+ R_{\sigma_0,\nu}(v) \Im\left( \frac{2^{-(p_0(z_0)+1) v}}{z_0 \sqrt{2p_0"(z_0)}} (1+ \delta_0(v))  \right)\right)
\end{eqnarray*}
where the $\delta_i$'s are functions which are of order $o(1)$ uniformly in $v$.

We clearly have $\sigma_1 = \sigma_0 - \frac{1}{m}$ and $z_1 = z_0 + O\left( \frac{1}{m} \right)$ and
therefore from the particular form of $R_{\sigma_i,\nu}(v)$  we deduce that
\begin{eqnarray*}
 R_{\sigma_1,\nu}(v) &= &R_{\sigma_0,\nu}(v) \left(1+ O(v^{-1/2})\right)\\
 \frac{1}{z_1 \sqrt{2p_1"(z_1)}} & = &  \frac{1}{z_0 \sqrt{2p_0"(z_0)}} \left(1 + O\left( \frac{1}{m} \right) \right)
\end{eqnarray*}
From this we deduce that
\begin{eqnarray}
\varepsilon_1 - \varepsilon_0 & =&  \frac{(-1)^{v}}{2 \binom{n-1}{w-1}} R_{\sigma_0,\nu}(v) \left( 
 \Im\left( \frac{2^{-p_1(z_1) v}}{z_0 \sqrt{2p_0"(z_0)}} (1+ o(1))  \right)
+  \Im\left( \frac{2^{-p_0(z_0) v}}{z_0 \sqrt{2p_0"(z_0)}} (1+ \delta_0(v))  \right)\right) \nonumber\\
&= &  \frac{(-1)^{v}}{2 \binom{n-1}{w-1}} R_{\sigma_0,\nu}(v) \Im \left( 
 \left( \frac{2^{-p_0(z_0) v}}{z_0 \sqrt{2p_0"(z_0)}} \left(1+ \delta_0(v)+ 2^{(p_0(z_0)-p_1(z_1)) v}(1+o(1))\right)  \right)
\right) \label{eq:e0_e1}
\end{eqnarray}

We now observe that
\begin{eqnarray}
p_0(z_0)-p_1(z_1) & = &  \log_{2}(z_0) - \frac{\sigma_0}{\nu} \log_{2}(1+z_0) - (\alpha - \frac{\sigma_0}{\nu}) \log_{2}(1-z_0)\\
&  &  - \log_{2}(z_1) + \frac{\sigma_1}{\nu} \log_{2}(1+z_1) + (\alpha - \frac{\sigma_1}{\nu}) \log_{2}(1-z_1)   \nonumber \\
&= &  \log_2 \frac{z_0}{z_1} - \frac{\sigma_0}{\nu} \log_{2}\frac{1+z_0}{1+z_1}  - (\alpha - \frac{\sigma_0}{\nu}) \log_{2}\frac{1-z_0}{1-z_1}
+ \left( \frac{\sigma_1}{\nu} -\frac{\sigma_0}{\nu}\right)\log_2\frac{1+z_1}{1-z_1}  \nonumber \\
& = & \log_2 \frac{z_0}{z_1} - \frac{\sigma_0}{\nu} \log_{2}\frac{1+z_0}{1+z_1}  - (\alpha - \frac{\sigma_0}{\nu}) \log_{2}\frac{1-z_0}{1-z_1}
- \frac{1}{v} \log_2\frac{1+z_1}{1-z_1} \label{eq:moyen}
\end{eqnarray}
where \eqref{eq:moyen} follows from the observation
\begin{eqnarray*}
\frac{\sigma_1}{\nu} - \frac{\sigma_0}{\nu}& = & \frac{t-1}{v} - \frac{t}{v}\\
& = & - \frac{1}{v}
\end{eqnarray*}

Recall that $z_1 = z_0 + \delta$ where $\delta = O(1/m)$  and that
\begin{eqnarray*}
\log_2 \frac{z_0}{z_1} - \frac{\sigma_0}{\nu} \log_{2}\frac{1+z_0}{1+z_1}  - (\alpha - \frac{\sigma_0}{\nu}) \log_{2}\frac{1-z_0}{1-z_1}& = &
p_0(z_0)-p_0(z_1) \\
& = & p_0(z_0)-p_0(z_0+\delta)
\end{eqnarray*}

The point is that $p_0'(z_0)=0$   and therefore 
$$
p_0(z_0)-p_0(z_0+\delta) = O(\delta^2) = O(1/m^2).
$$

Using this in \eqref{eq:moyen}  and then multiply by $v$ 
  implies 
\begin{equation}
(p_0(z_0)-p_1(z_1))v =  - \log_2 \frac{1-z_1}{1+z_1}+O(1/v)=  - \log_2 \frac{1-z_0}{1+z_0}+O(1/v)
\end{equation}
We can substitute for this expression in \eqref{eq:e0_e1} and obtain
\begin{eqnarray}
\varepsilon_1 - \varepsilon_0 
&= &  \frac{(-1)^{v}}{2 \binom{n-1}{w-1}} R_{\sigma_0,\nu}(v) \Im \left( 
 \left( \frac{2^{-p_0(z_0) v}}{z_0 \sqrt{2p_0"(z_0)}} \left(1 + \frac{1+z_0}{1-z_0} + o(1) \right)  \right)
\right) \nonumber \\
& = & 
\frac{(-1)^{v}}{\binom{m}{v}} R_{\sigma_0,\nu}(v) \Im \left( 
 \left( \frac{2^{-p_0(z_0) v}}{(z_0-z_0^2) \sqrt{2p_0"(z_0)}} (1+ o(1))  \right)
\right) \label{eq:ouf}
\end{eqnarray}
Recall that 
\begin{eqnarray*}
R_{\sigma,\nu}(v) & = & \frac{1+o(1)}{\sqrt{\pi v}}\\
\binom{m}{v} & = & \frac{2^{mH(\nu)}}{\sqrt{2 \pi m \nu (1-\nu)}} \\
& = & \frac{2^{\frac{vH(\nu)}{\nu}}}{\sqrt{2 \pi v  (1-\nu)}}
\end{eqnarray*}
By using this in \eqref{eq:ouf} we obtain
\begin{eqnarray}
\varepsilon_1 - \varepsilon_0 
& = & 
\frac{(-1)^v \sqrt{2 \pi v (1-\nu)}}{\sqrt{ \pi v}\left|(z_0-z_0^2)\sqrt{2p_0"(z_0)}\right|} 2^{-v \left( \Re(p_0(z_0))+ \frac{H(\nu)}{\nu}  \right)}
\sin \left( v \theta - \theta_0 \right) (1+o(1)) \\
& = & \frac{(-1)^v \sqrt{1-\nu}}{\left|(z_0-z_0^2)\sqrt{p_0"(z_0)}\right|} 2^{-v \left( \Re(p_0(z_0))+ \frac{H(\nu)}{\nu}  \right)}
\sin \left( v \theta - \theta_0 +o(1)\right) (1+o(1))
\end{eqnarray}
\end{proof}

From Lemmas \ref{lem:realcomplete} and \ref{lem:complex} we deduce immediately that
	
	\begin{corollary} 
		\label{cor:biasSDecoding}
		
	We set $\gamma = \frac{1}{\omega}$,

		\begin{itemize} 
			
			\item If $\frac{\tau}{\omega} \in ( 0,\gamma/2 - \sqrt{\gamma - 1})$:
			\[ \pwt = 2  \left(  \omega  \left( \log_{2}(r) - \frac{\tau}{\omega} \log_{2}(1-r) - (\gamma - \frac{\tau}{\omega}) \log_{2}(1+r) \right)  +H(\omega) \right) \]
			\[ \mbox{where } r \mbox{ is the smallest root of } (\gamma - 1) X^{2} - (\gamma - 2\frac{\tau}{\omega}) X + 1  \]

			\item  If $\frac{\tau}{\omega} \in ( \gamma/2-\sqrt{\gamma-1},\gamma/2 )$:
			\[ \pwt = 2  \left( \omega  \Re \left( \log_{2}(z) -\frac{\tau}{\omega} \log_{2}(1+z) - (\gamma - \frac{\tau}{\omega}) \log_{2}(1-z) \right) + H\left( \omega \right)\right)  \]
			\[ \mbox{ where } z = r e^{\ic \varphi} \mbox{ with } r = \frac{1}{\sqrt{\gamma-1}} \mbox{ and } \cos(\varphi) = \frac{ 2\frac{\tau}{\omega} - \gamma}{2\sqrt{\gamma - 1}}  \]
		\end{itemize}
	\end{corollary}
	
\begin{rem}
These asymptotic formulas turn out to be already accurate in the "cryptographic range" as it is shown in Figure \ref{fig:numBias}.
\end{rem}

					\begin{figure}[h!]
					\centering
					\includegraphics[scale = 0.6]{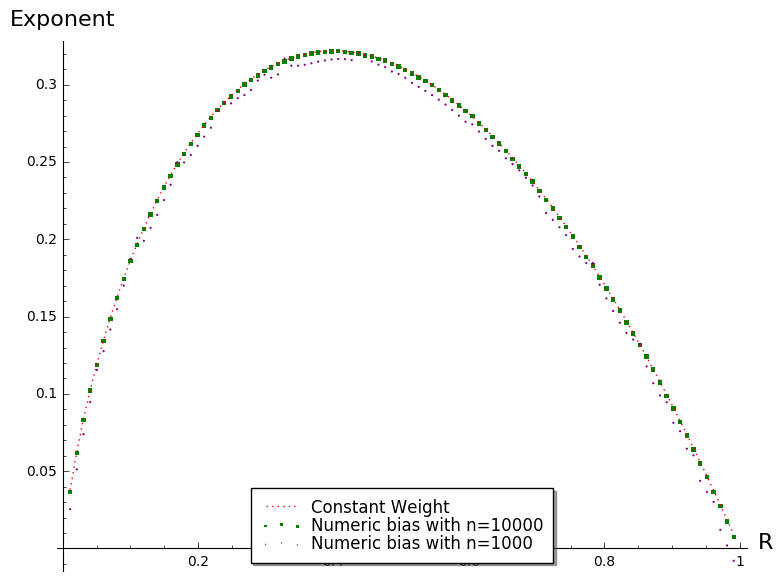}
					\caption{Comparison of the asymptotic and numeric exponents for $\tau = H^{-1}(1-R)$. \label{fig:numBias}}
				\end{figure}

Amazingly enough these formulas can be simplified a lot in the second case of the corollary as shown by the following theorem.
\begin{theorem}[Asymptotic complexity of statistical decoding]
\label{biasSDecoding}
\mbox{ }\\
		\begin{itemize} 
			
			\item If $\tau \in \left( 0,\frac{1}{2}  - \sqrt{\omega - \omega^2}\right)$: $\pwt = 2   \omega   \log_{2}(r) - 2 \tau \log_{2}(1-r) - 2(1- \tau) \log_{2}(1+r)   + 2H(\omega)$ where $r$ is the smallest root of  $(1- \omega) X^{2} - (1 - 2\tau) X + \omega =0$.

			\item  If $\tau \in \left( \frac{1}{2}  - \sqrt{\omega - \omega^2},\frac{1}{2} \right)$: $\pwt =  H(\omega)+H(\tau) -1.$
		\end{itemize}
\end{theorem}
\begin{proof}
The first case is just a slight rewriting. To prove the formula corresponding to the second case let us recall that 
the $z$ that appears in the second case of Corollary \ref{cor:biasSDecoding} satisfies $p'(z)=0$ where
$$
p(z) \eqdef \omega \log_2  z - \tau \log_2 (1+z) - (1 - \tau) \log_2(1-z).
$$
Let 
\begin{eqnarray*}
f(\omega,\tau) &\eqdef &2  \left(  \omega  \Re \left( \log_{2}(z) - \frac{\tau}{\omega} \log_{2}(1+z) - (\gamma - \frac{\tau}{\omega}) \log_{2}(1-z) \right)  +H(\omega)
\right)\\
& = & 2 \Re(p(z)) + 2H(\omega).
\end{eqnarray*}
Let us first differentiate this expression with respect to 
$\omega$:
\begin{eqnarray}
\frac{\partial f(\omega,\tau)}{\partial \omega} & = & 2 \Re(p'(z))\frac{\partial z}{\partial \omega} + 2 \Re(\log_2(z)) + 2\log_2 \frac{1-\omega}{\omega} \nonumber \\ 
& = & 2 \Re(\log_2(z)) + 2 \log_2 \left(\frac{1-\omega}{\omega}\right) \label{eq:derivative1}
\end{eqnarray}
Since $z = r e^{\ic \varphi}$  with  $r = \frac{1}{\sqrt{\gamma-1}}$, we deduce that
$$
2 \Re(\log_2(z)) = 2\log_2 r = 2 \log_2\left(\frac{1}{\sqrt{\gamma-1}}\right) = \log_2\left( \frac{1}{1/\omega-1}\right) = \log_2\left(\frac{\omega}{1-\omega}\right).
$$
Substituting this expression for $2 \Re(\log_2(z))$ in \eqref{eq:derivative1} yields
\begin{equation}\label{eq:partial_omega}
\frac{\partial f(\omega,\tau)}{\partial \omega} =   \log_2\left(\frac{\omega}{1-\omega}\right) + 2 \log_2 \left(\frac{1-\omega}{\omega} \right)=  \log_2 \left(\frac{1-\omega}{\omega}\right) = H'(\omega).
\end{equation}
We continue the proof by differentiating now $f(\omega,\tau)$ with respect to $\tau$:
\begin{eqnarray}
\frac{\partial f(\omega,\tau)}{\partial \tau} & = & 2 \Re(p'(z))\frac{\partial z}{\partial \tau} - 2 \Re\left(\log_2(1+z) - \log_2(1-z) \right)  \nonumber \\ 
& = & - 2 \Re\left(\log_2\left(\frac{1+z}{1-z}\right)\right)  \nonumber
\end{eqnarray}
Recall that $z$ is also given by one of the two roots of $(1-\omega)X^2 + (1-2\tau) X + \omega = 0$ (see Theorem \ref{th:expansion} for the root which is actually chosen) 
and therefore
$$
z = \frac{2\tau -1 + \ic \sqrt{4 \omega(1-\omega)-(1-2\tau)^2}}{2(1-\omega)}
$$
From this we deduce that
\begin{eqnarray*}
1 + z & = & \frac{1 -2\omega + 2\tau + \ic \sqrt{4 \omega(1-\omega)-(1-2\tau)^2}}{2(1-\omega)}\\
1-z & = & \frac{3 -2\omega - 2\tau - \ic \sqrt{4 \omega(1-\omega)-(1-2\tau)^2}}{2(1-\omega)}\\
\end{eqnarray*}
\begin{eqnarray*}
- 2 \Re  \left(\log_2\left(\frac{1+z}{1-z}\right)\right)  & = & - 2 \Re  \left(\log_2\left(\frac{1 -2\omega + 2\tau + \ic \sqrt{4 \omega(1-\omega)-(1-2\tau)^2}}{3 -2\omega - 2\tau - \ic \sqrt{4 \omega(1-\omega)-(1-2\tau)^2}}\right)\right)\\
& = &- 2 \log_2 \left| \frac{1 -2\omega + 2\tau + \ic \sqrt{4 \omega(1-\omega)-(1-2\tau)^2}}{3 -2\omega - 2\tau - \ic \sqrt{4 \omega(1-\omega)-(1-2\tau)^2}}\right|\\
& = &- \log_2 \frac{(1 -2\omega + 2\tau)^2+4 \omega(1-\omega)-(1-2\tau)^2 }{(3 -2\omega - 2\tau)^2+4 \omega(1-\omega)-(1-2\tau)^2 }\\
& = & -\log_2 \frac{1+4\omega^2+4\tau^2-4\omega+4\tau-8\omega \tau+4\omega-4\omega^2-1-4\tau^2+4\tau }{9+4\omega^2+4\tau^2-12\omega-12\tau+ 8\omega \tau+4\omega-4\omega^2-1-4\tau^2+4\tau}\\
& = & -\log_2 \frac{8\tau -8 \omega \tau}{8-8\omega-8\tau+8 \omega \tau}\\
& = & -\log_2 \frac{8\tau(1-\omega)}{8(1-\omega)(1-\tau)}\\
& = & - \log_2 \frac{\tau}{1-\tau}\\
& = & H'(\tau)
\end{eqnarray*}
These two results on the derivative imply that
$$
f(\omega,\tau) = H(\omega) + H(\tau) + C
$$
for some constant $C$ which is easily seen to be equal to $-1$ by letting $\omega$ go to $0$ and $\tau$ go to $\frac{1}{2}$ in $f(\omega,\tau)$.

\end{proof}

\section{The binomial model}

\cite{FKI07} introduced another model for the parity-check equations used in statistical decoding. 
Instead of assuming that they are chosen randomly of a given weight $w$, the authors of \cite{FKI07} assume that they are random binary words of length $n$ 
where the entries  are chosen independently of each other according to a Bernoulli distribution of parameter $w/n$. In other words, the expected weight 
is still $w$ but the weight of the parity-check equation is not fixed anymore and may vary. We will call it the {\em binomial model} of
weight $w$ and length $n$ and refer to our model  as the constant weight model of weight $w$.
The binomial model presents the advantage of simplifying significantly 
the analysis of statistical decoding. It is easy to analyze the simple statistical decoding algorithm that we consider here and to compute asymptotically the
number of parity-check equations that ensure successful decoding. We will do this in what follows. 
But the authors of \cite{FKI07} went further since they were even able to analyze asymptotically
an iterative version of statistical decoding by following some of the ideas of \cite{SV04}.
They showed that
	\begin{proposition}[{\cite[Proposition 2.1 p.405]{FKI07}}]
        In the binomial model of weight $w$ and length $n$, the number of check sums that are necessary to correct with large enough probability
$t$ errors by using the iterative decoding algorithm of \cite{FKI07}   is  well estimated by $O(J_{\text{min}})$ with 
	\[ J_{\text{min}} = \left( \frac{n}{n-2w} \right)^{2(t-1)} = \left( 1- \frac{2w}{n} \right)^{-2(t-1)} \]
	where the constant in the ``big O'' depends on the ratio $t/n$. 
	\end{proposition} 
 Let us first show that naive statistical decoding performs almost as well when we forget about polynomial factors. It makes sense in order to compare both models
to introduce some additional notation.
\begin{eqnarray*}
\qbin_0 & = & \pbin \left( \langle e,h \rangle = 1 |h_i=1\right) \mbox{ when } e_{i} = 0\\
\qbin_1 & = & \pbin \left( \langle e,h \rangle = 1 |h_i=1\right) \mbox{ when } e_{i} = 1\\
\end{eqnarray*}
where $h$ is a parity-check equation chosen according to the binomial model and the probability is taken over the random choice of $h$
in this model (and $\pbin$ means that we take the probabilities according to the binomial model).
These quantities do not depend on $i$. It will also be convenient to define 
$\epsbin_0$ and $\epsbin_0$ as
$$
\qbin_0 = \frac{1}{2} + \epsbin_0 \;\;; \;\; \qbin_1 = \frac{1}{2} + \epsbin_1.
$$

The computations of \cite[Sec II. B]{FKI07} show that
\begin{eqnarray*}
\qbin_0 & = & \frac{1 - \left(1-\frac{2w}{n}\right)^t}{2}\\
\qbin_1 & = & \frac{1 + \left(1-\frac{2w}{n}\right)^{t-1}}{2}
\end{eqnarray*}
This implies that
$$
\epsbin_0 = - \frac{\left(1-\frac{2w}{n}\right)^t}{2} \;\;;\;\; \epsbin_1 =  \frac{\left(1-\frac{2w}{n}\right)^{t-1}}{2}.
$$

It is also convenient in order to distinguish both models  to rename the quantities $q_0$, $q_1$, $\varepsilon_0$ and $\varepsilon_1$ that were introduced before by 
referring to them as $\qcon_0$, $\qcon_1$, $\epscon_0$ and $\epscon_1$ respectively.
We can perform the same statistical test as before by computing from $m$ parity-check equations $h^1,\dots,h^m$ all involving the bit $i$ we want to decode, the quantity
$$
V_m = \sum_{k=1}^m \sgn(\epsbin_1 -\epsbin_0) \langle y,h^{k}\rangle= \sum_{k=1}^m \langle y,h^{k}\rangle.
$$
The expectation of this quantity is $E_b \eqdef m\left( \frac{1}{2} + \epsbin_b \right)$ depending
on the value $b\in \{0,1\}$ of the bit we want to decode.
We decide that the bit we want to decode is equal to $0$ if $V_m < \frac{E_0+E_1}{2}$ and $1$ otherwise. As before, we observe that by Chernoff's bound we make a wrong decision with probability at most $2 \cdot 2^{-m \frac{(\epsbin_1- \epsbin_0)^2}{2 \ln(2)}}$.
This probability can be made to be of order $o(1/n)$ by choosing $m$ as
$m = K \log n \frac{1}{(\epsbin_1- \epsbin_0)^2}$ for a suitable constant $K$. In this case, decoding the whole sequence succeeds with probability $1-o(1)$.
In other words, naive statistical decoding succeeds for $m = O\left( \log n \frac{1}{(\epsbin_1- \epsbin_0)^2}\right)$.

We may observe now that 
\begin{eqnarray*}
\frac{1}{(\epsbin_1- \epsbin_0)^2} & = & O\left( (1-2w/n)^{-2{t-1}}\right)\\
& = & O(J_{\text{min}})
\end{eqnarray*}
This means that naive statistical decoding needs only marginally more equations in the binomial model (namely a multiplicative factor of order $O(\log n)$).
To summarize the whole discussion, 
the number of parity-checks needed for decoding is
\begin{itemize}
\item with iterative statistical decoding over the binomial model
$$O\left(\frac{1}{(\epsbin_1- \epsbin_0)^2}\right),$$
\item with naive statistical decoding over the binomial model
$$O\left(\frac{\log n}{(\epsbin_1- \epsbin_0)^2}\right)$$
\item with naive statistical decoding over the constant weight model
$$O\left(\frac{\log n}{(\epscon_1- \epscon_0)^2}\right).$$
\end{itemize}

One might wonder now whether there is a difference between both models. It is very tempting to conjecture that both models are very close to each other since the expected weight of the parity-checks is $w$ in both cases. However this is not the
case, we are really in a large deviation situation where the bias of some extreme weights take over the bias corresponding to the typical weight of the parity check equations.  
To illustrate this point, we choose the weight to be $w = \omega n$, the number of errors as $t= \tau n$ for some fixed $\omega$ and $\tau$, and then  let $n$ go to infinity.
The normalized exponent\footnote{Here the number of equations is a function of the form
$\tilde{O}\left(e^{\alpha(\tau,\omega)n}\right)$ and we mean here the coefficient $\alpha(\omega,\tau)$.} 
of the number of parity-check equations which is needed  is
$$
\lim\limits_{n \to +\infty} \frac{1}{n}  \log_{2}\left(\frac{1}{(\epsbin_1 - \epsbin_0)^{2}}\right) = -2 \tau  \log_{2}\left(1- 2 \omega\right)
$$
in the binomial case, whereas 
$
\lim\limits_{n \to +\infty} \frac{1}{n}  \log_{2}\left(\frac{1}{(\epscon_1 - \epscon_0)^{2}}\right)
$ 
is given by Theorem \ref{biasSDecoding} in the constant weight case and both terms are indeed different in general. One case which is particularly interesting is when
$\tau$ and $\omega$ are chosen as  $\tau = H^{-1}(1-R)$ and $\omega = R/2$, where $R$ is the code rate we consider. 
This corresponds to the hardest case of syndrome decoding and when the parity-check equations of this weight can
be easily obtained as we will see in Section \ref{sec:naive}.
The  two normalized exponents  are compared on Figure \ref{fig:KF} as a function of the rate $R$.  As we see, there is a huge difference.
The problem with the model chosen in \cite{FKI07}  is that it is a very favorable model for statistical decoding. To the best of our knowledge there are no 
	efficient algorithms for producing such parity-checks when $\omega \leq R/2$. 
	Note that even such an algorithm were to exist, selecting appropriately only one weight would not change the exponential complexity of the algorithm
(this will be proved in Section \ref{sec:single}). In other words, in order to study statistical decoding we may restrict ourselves, as we do here, to
considering only one weight and not a whole range of weights.

			\begin{figure}[h!]
			\centering
			\includegraphics[scale = 0.6]{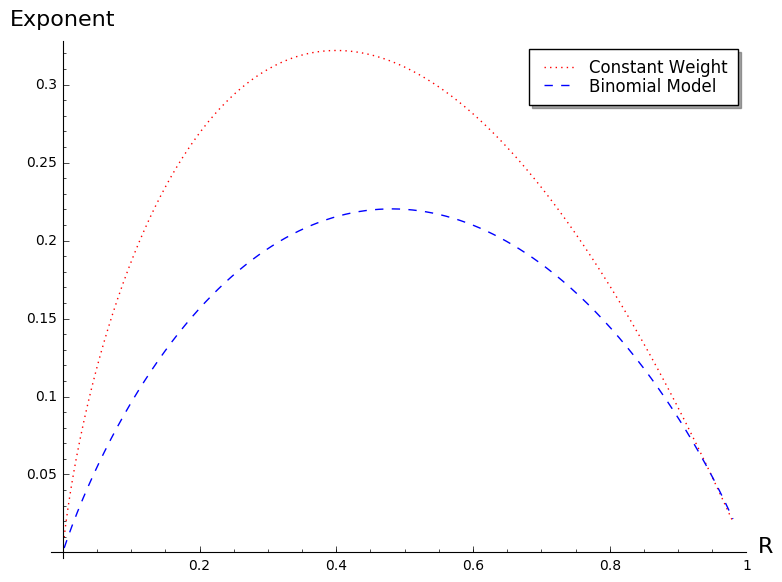}
			\caption{Comparison of the normalized exponents with $\tau = H^{-1}(1-R)$ of the number of parity-check equations which are needed in the binomial and the constant weight case. \label{fig:KF}}
		\end{figure}
		The difference between both formulas is even more apparent when considering the slopes at the origin as shown in Figure \ref{fig:KFS}.
		\begin{figure}
			\centering
			\includegraphics[scale = 0.6]{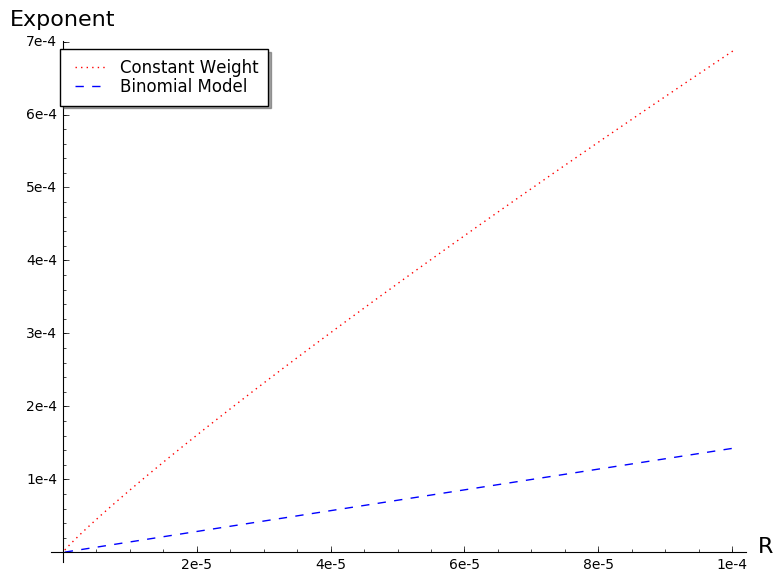}
			\caption{Comparison of the complexities with $\tau = H^{-1}(1-R)$ for rate close to $0$ \label{fig:KFS}}
		\end{figure} 
However both models get closer when the error weight decreases. For instance when considering a relative error $\tau=H^{-1}(1-R)/2$, we see in Figure \ref{fig:KFDGV2} that the difference between both models
gets significantly smaller. Actually the difference vanishes when the relative error tends to $0$, as shown by Proposition \ref{subcpx}.

					\begin{figure}[h!]
					\centering
					\includegraphics[scale = 0.6]{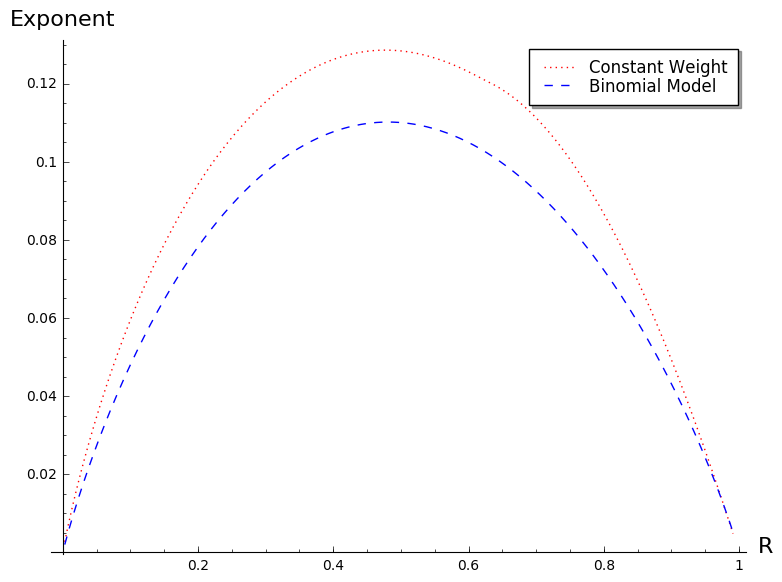}
					\caption{Comparison of the normalized exponents with $\tau = H^{-1}(1-R)/2$ of the number of parity-check equations which are needed 
in the binomial and the constant weight case. \label{fig:KFDGV2}}
				\end{figure}

		\begin{proposition} [Asymptotic complexity of statistical decoding for a sub-linear error weight]\label{subcpx}$ $ 
		\[ \pi(\omega,\tau) \mathop{=}\limits_{\tau \rightarrow 0} -2\tau \log_{2}(1-2\omega) + o(\tau) \] 
		\end{proposition}

		\begin{proof}
		As $\tau$ decreases to $0$, we consider for $\pi(\omega,\tau)$ the first formula which is given in Theorem \ref{biasSDecoding}. We have:
		\begin{align}
		\pi(\omega,\tau) &= 2 \omega \log_{2}(r) - 2 \tau \log_{2}(1-r) - 2(1 - \tau) \log_{2}(1+r) + 2H(\omega) \nonumber \\ 
		&= 2 \omega \log_{2}(r) - 2 \log_{2}(1+r) -2 \tau \log_{2}\left( \frac{1-r}{1+r} \right) + 2 H (\omega) \label{eq:a} 
		\end{align}
		with
		\[ r = \frac{1-2\tau - \sqrt{ (1 - 2 \tau)^{2} - 4 \omega (1-\omega)}}{2(1-w)} \]

		Let us compute now  Taylor series expansion of $r$ when $\tau \rightarrow 0$. We start with
		\begin{align*}
		r&= \frac{1-2\tau - \sqrt{ 1 - 4 \omega (1-\omega) -4\tau + 4\tau^{2}}}{2(1-\omega)} \\
		&= \frac{1-2\tau - \sqrt{ (1 - 2 \omega)^{2} - 4 \tau + o(\tau) }}{2(1-\omega)}
		\end{align*}

		Now using the fact that:
		\[ (A^{2} - \varepsilon)^{1/2} \mathop{=}\limits_{\varepsilon \rightarrow 0} A- \frac{\varepsilon}{2A} + o(\varepsilon)  \] 
		we have:
		\begin{align*}
		r&= \frac{1-2\tau - (1-2\omega) + \frac{2\tau}{1-2\omega} + o(\tau)}{2(1-\omega)} \\
		&=  \frac{\omega}{1- \omega} + \frac{ \tau \omega}{(1-\omega)(1-2\omega)} +o(\tau)
		\end{align*}

		And we deduce that:
		\[ 1-r =  \frac{1-2\omega}{1-\omega}- \frac{ \tau \omega}{(1-\omega)(1-2\omega)} +o(\tau) \] 
		\[1+r =  \frac{1}{1-\omega} +  \frac{ \tau \omega}{(1-\omega)(1-2\omega)} +o(\tau) \]

		and therefore
		\begin{equation}
		\label{eq:b} 
		-2 \tau \log_{2}\left( \frac{1-r}{1+r} \right) \mathop{=}\limits_{\tau \rightarrow 0}  -2 \tau \log_{2}(1-2\omega) + o(\tau)
		\end{equation}

		Now using the fact that:
		\[ \log_{2}(A + \varepsilon) \mathop{=}\limits_{\varepsilon \rightarrow 0} \log_{2}(A) + (1/\ln(2)) \left( \frac{\varepsilon}{A} + o(\varepsilon) \right)  \]
		we have the asymptotic expansions with the logarithms:
		\[\log_{2}(r) = \log_{2}\left( \frac{\omega}{1-\omega} \right) + (1/\ln(2)) \left( \frac{\tau}{1 - 2 \omega} + o(\tau) \right)  \]
		\[ \log_{2}(1+r) = \log_{2}\left( \frac{1}{1-\omega} \right) + (1/\ln(2)) \left( \frac{\tau \omega}{1 - 2 \omega} + o(\tau) \right) \]
		So we deduce that:
		\begin{align*}
		2 \omega \log_{2}(r) - 2 \log_{2}(1+r) &= 2 \omega  \log_{2}\left( \frac{\omega}{1-\omega} \right) - 2 \log_{2}\left( \frac{1}{1-\omega} \right) + o(\tau) \\
		&= 2\omega \log_{2}(\omega) + 2 (1-\omega) \log_{2}(1-\omega) + o(\tau) \\
		&= -2 H(\omega) + o(\tau)
		\end{align*}

		So by plugging this expression with \eqref{eq:b} in \eqref{eq:a} we have the result.

		\end{proof} 
		
		The sublinear case is also relevant to cryptography since several 
McEliece cryptosystems actually operate at this regime, this is true for the original McEliece system with fixed rate binary Goppa codes \cite{M78} or with the 
MDPC-McEliece cryptosystem \cite{MTSB13}. In this regime, \cite{CS16} showed that all ISD algorithms have the same asymptotic complexity when the number $t$ of errors to correct is equal to $o(n)$ and this is given by:
		\[ 2^{-t \log_{2}(1-R)(1+  o(1))} \]

		Let us compare the exponents of statistical decoding and the ISD algorithms when we want to correct a sub-linear error weight. 
When $t=o(n)$ the complexity we are after is subsexponential in the length.
The only algorithm finding moderate weight parity-check equations in subexponential time we found is Algorithm \ref{alg:gauss}. It produces parity-check equations of weight $Rn/2$ in amortized time $\tilde{O}(1)$. So with this algorithm, the exponent of statistical decoding is given by $-2\tau \log_{2}(1-R)$ which is twice the exponent of all the ISDs. We did not conclude for a relative weight $< R/2$ as in any case,
all the algorithms we found  needed exponential time to output enough equations to perform statistical decoding. So unless one comes up with 
an algorithm that is able to produce parity-check equations of relative weight $< R/2$ in subexponential time, statistical decoding is not better that any ISDs when we have to correct $t=o(n)$ errors.

\section{Studying the single weight case is sufficient}
\label{sec:single}

The previous section showed that if it is much more favorable when it comes to perform statistical decoding to produce parity-check equations following the binomial model of weight $w$
rather than parity-checks of constant weight $w$. The problem is that as far as we know, there is no efficient  way of producing 
moderate weight parity-check equations (let us say that we call moderate any weight $\leq 1+Rn/2$)  which would follow such a model. Even the ``easy case'', where $w = 1+Rn/2$ and where it is trivial to
produce such equations by simply putting the parity-check matrix in systematic form and taking rows in this matrix
\footnote{For more details 
see Section \ref{sec:naive}}, does not follow the binomial model  : the standard deviation 
of the parity-check equation weight is easily seen to be different between what is actually produced by the algorithm and the binomial model of weight $1+Rn/2$. Of course, this does not mean that 
we should rule out the possibility that there might exist such efficient algorithms. We will however prove  that under very mild conditions, that even such an algorithm were to exist
then anyway it would produce by nature parity-checks of different weights and that we would have a statistical decoding algorithm of the same exponential complexity which would keep only {\em one very 
specific weight}. In other words, it is sufficient to care about the single weight case as we do here when we study just the exponential complexity of statistical decoding.

To verify this, we fix an arbitrary position we want to decode and assume that some algorithm has produced in time $T$,
  $m = \sum_{j=1}^n m_j$
parity check equations involving this position where $m_j$ denotes the number of parity-check equations of weight $j$.                       
The equations of weight $j$ are denoted by $h_{1}^{j},\dots,h_{m_j}^j$. Statistical decoding is based on simple statistics involving
the values
 $\langle y,h_{s}^{j}\rangle$. To simplify a little bit the expressions we are going to manipulate, let us introduce
	\[ X_{s}^{j} \eqdef \langle y,h_{s}^{j} \rangle \]

	Similarly to Assumptions \ref{ass:one} and \ref{ass:two}, we assume that the distribution of $\langle y,h_{s}^{j} \rangle$ is approximated by the distribution of $\langle y,h_{s}^{j} \rangle$ when $h_{s}^{j}$ is drawn uniformly at random among the words of weight $j$ and the $\langle y,h_{s}^{j} \rangle$'s are independent. So we have $X_{s}^{j} \sim \mathcal{B}(1/2 + \varepsilon_{l}(j))$ under the hypothesis $\mathcal{H}_{l}$ and $\varepsilon_{l}(j)$ is the bias defined in Subsection \ref{bias} for a weight $j$. Our aim now is to find a test distinguishing both hypotheses $\mathcal{H}_{0}$ and $\mathcal{H}_{1}$. As in Subsection \ref{bias} it will be the Neymann-Pearson test. 
We define the following quantity where $\mathbb{P}_{\mathcal{H}_{l}}$ denotes the probability under the hypothesis $\mathcal{H}_{l}$: 
	\[  q \eqdef \ln \left( \frac{\mathbb{P}_{\mathcal{H}_{0}} \left( X_{1}^{1} = x_{1}^{1},\cdots,X_{m_{1}}^{1}= x_{m_1}^{1},\cdots, X_{1}^{n} = x_{1}^{n},\cdots,X_{m_{n}}^{n}= x_{m_n}^{n} \right) }{\mathbb{P}_{\mathcal{H}_{1}} \left(  X_{1}^{1} = x_{1}^{1},\cdots,X_{m_{1}}^{1}= x_{m_1}^{1},\cdots, X_{1}^{n} = x_{1}^{n},\cdots,X_{m_{n}}^{n}= x_{m_n}^{n} \right)} \right)   \]

	The lemma of Neymann-Pearson tells to us to proceed as follows: if $q>\Theta$, where $\Theta$ is some threshold, choose $\mathcal{H}_{0}$ and $\mathcal{H}_{1}$ otherwise. In this case, no other statistic test will lead to lower false detection probabilities at the same time. In our case, it is enough to set the threshold $\Theta$ to $0$ since it can be easily verified that no other choices will not change the exponent of the number of samples we need for having vanishing false detection probabilities.
	We set $p_{l}(j) \eqdef 1/2 + \varepsilon_{l}(j)$, $I_{0}(j) = \# \{ 0 \in \{x_{1}^{j},\cdots,x_{m_{j}}^{j} \} \}$ and $I_{1}(j) = \# \{ 1 \in \{x_{1}^{j},\cdots,x_{m_{j}}^{j} \} \}$, we have:
 	\begin{align*}
 		\frac{\mathbb{P}_{\mathcal{H}_{0}} \left( X_{1}^{1} = x_{1}^{1},\cdots,X_{m_{1}}^{1}= x_{m_1}^{1},\cdots, X_{1}^{n} = x_{1}^{n},\cdots,X_{m_{n}}^{n} = x_{m_n}^{n}\right) }{\mathbb{P}_{\mathcal{H}_{1}} \left(  X_{1}^{1} = x_{1}^{1},\cdots,X_{m_{1}}^{1}= x_{m_1}^{1},\cdots, X_{1}^{n} = x_{1}^{n},\cdots,X_{m_{n}}^{n}= x_{m_n}^{n} \right)} &= \prod_{j=1}^{n} \frac{ p_{0}(j)^{I_{1}(j)}\cdot (1-p_{0}(j))^{I_{0}(j)}}{ p_{1}(j)^{I_{1}(j)}\cdot (1-p_{1}(j))^{I_{0}(j)}} 
 		\end{align*}

 	Therefore by taking the natural logarithm of this expression and  $\sum_{k=1}^{m_{j}} X_{k}^{j} = I_{1}(j)$ and $I_{1}(j)+ I_{0}(j) = m_{j}$, we have:
 		\begin{align*}
 		q &= \sum_{j=1}^{n} I_{0}(j)  \left[ \ln(1-p_{0}(j)) - \ln(1-p_{1}(j)) \right] + I_{1}(j)  \left[ \ln(p_{0}(j)) - \ln(p_{1}(j)) \right] \\
 		&= \sum_{j=1}^{n} (m_{j} - I_{1}(j))  \left[ \ln(1-p_{0}(j)) - \ln(1-p_{1}(j)) \right] + \sum_{s=1}^{m_{j}} X_{s}^{j}   \left[ \ln(p_{0}(j)) - \ln(p_{1}(j)) \right] \\
 		&= \sum_{j=1}^{n} \sum_{s=1}^{m_{j}} X_{s}  \left[ \ln(p_{0}(j)) - \ln(1-p_{0}(j)) + \ln(1-p_{1}(j)) - \ln(p_{1}(j)) \right] \\
 		& \qquad \qquad + m_{j} \ln \frac{1-p_{0}(j)}{1-p_{1}(j)}
 		\end{align*}

 		We now use the Taylor series expansion around $0$ : $\ln(1/2 + x) = -\ln(2) + 2x - \frac{4x^{2}}{2} + \frac{8x^{3}}{3} + o(x^{3})$ and we deduce
		for $i$ in $\{0,1\}$: 
 		\begin{align*}
 		\ln(p_{i}(j)) &= \ln(1/2 + \varepsilon_{i}(j)) \\
 		&=-\ln(2) + 2\varepsilon_{i}(j) - 2\varepsilon_{i}(j)^2 + (8/3) \varepsilon_{i}(j)^3 + o(\varepsilon_{i}(j)^3) 
 		\end{align*} 
 		\begin{align*}
 		\ln(1-p_{i}(j)) &= \ln(1/2 - \varepsilon_{i}(j)) \\
 		&=-\ln(2) -2\varepsilon_{i}(j) - 2\varepsilon_{i}(j)^2 - (8/3) \varepsilon_{i}(j)^3 + o(\varepsilon_{i}(j)^3) 
 	\end{align*} 	
 		
 		We have,
 		\begin{align*}
 		q &= \sum_{j=1}^{n} \sum_{s=1}^{m_{j}} X_{s} \cdot \left( (4\varepsilon_{0}(j) + (16/3) \varepsilon_{0}(j)^{3} + o(\varepsilon_{0}(j)^{3}) - 4 \varepsilon_{1}(j) - (16/3) \varepsilon_{1}(j)^{3} + o(\varepsilon_{1}(j)^{3}) \right) \\
 		&\qquad \qquad  - 2 m_{j} \cdot \left( \varepsilon_{0}(j) - \varepsilon_{1}(j) +o(\varepsilon_{0}(j)) + o(\varepsilon_{1}(j)  ) \right) \\
 		&= 4 \sum_{j=1}^{n} \sum_{s=1}^{m_{j}} X_{s}^{j} \left( ( \varepsilon_{0}(j) - \varepsilon_{1}(j) + O(\varepsilon_0(j)^3)
		+   O(\varepsilon_1(j)^3) \right)\\
 		& \qquad \qquad + m_{j} \ln \frac{1-p_{0}(j)}{1-p_{1}(j)} \\
 		&\approx  4 \sum_{j=1}^{n} \sum_{s=1}^{m_{j}}  Y_{s}^{j} + C 
 		\end{align*} 
 	where
	\[Y_{s}^{j} \eqdef (\varepsilon_{0}(j) - \varepsilon_{1}(j)) X^{j}_{s} \] 
	and $C$ is the constant defined by:
	\[ C \eqdef  + m_{j} \ln \frac{1-p_{0}(j)}{1-p_{1}(j)} \] 
%
	This computation suggests to use the random variables $Y_{s}^{j}$ to build our distinguisher with the Neyman-Pearson likelihood test. By the assumptions on the $X_{s}^{j}$'s, the $Y_{s}^{j}$'s are independent and we have under $\mathcal{H}_{l}$:
	\[ \mathbb{P} \left( Y_{s}^{j} =0 \right) =  \frac{1}{2} - \varepsilon_{l}(j) \quad ; \quad \mathbb{P} \left( Y_{s}^{j} =  (\varepsilon_{0}(j) - \varepsilon_{1}(j)) \right) =  \frac{1}{2} + \varepsilon_{l}(j) \]

	The expectation of $Y_{s}^{j}$ under $\mathcal{H}_{l}$ is given by:
	\[\mathbb{E}\left( Y_{s}^{j} \right) =  (\varepsilon_{0}(j) - \varepsilon_{1}(j))\cdot \left( \frac{1}{2} + \varepsilon_{l}(j) \right)  \]

	As for our previous distinguisher we define the random variable $V_{m}$ for $m=\sum_{j=1}^{n} m_{j}$ uniform and independent draws of vectors $h_{s}^{j}$ in $\sH_{w_{j},i}$: 
	\[ V_{m} \eqdef \sum_{j=1}^{n} \sum_{s=1}^{m_{j}} Y_{s}^{j} \]

	The expectation  of $V_{m}$ depends on which hypothesis $\mathcal{H}_{l}$ holds. When hypothesis $\mathcal{H}_{l}$ holds,
	we denote the expecation of $V_n$ by $E_l$. The difference $E_{0} - E_{1}$ is given by:
	\begin{align*}
E_{0} - E_{1} &=	\sum_{j=1}^{n} \sum_{k=1}^{m_{j}} (\varepsilon_{0}(j) - \varepsilon_{1}(j)) \left( \frac{1}{2} + \varepsilon_{0}(j) \right) -  (\varepsilon_{0}(j) - \varepsilon_{1}(j)) \left( \frac{1}{2} + \varepsilon_{1}(j) \right) \\
	&= \sum_{j=1}^{n} \sum_{k=1}^{m_{j}} (\varepsilon_{0}(j) - \varepsilon_{1}(j))^{2}  \\
	&= \sum_{j=1}^{n} m_{j} (\varepsilon_{0}(j) - \varepsilon_{1}(j))^{2}
	\end{align*}

	The deviations of $V_m$ around its expectation will be quantified through Hoeffding's bound which gives in this case up to constant factors
	in the exponent the right behavior of the probability that $V_m$ deviates from its expectation
	\begin{proposition}[Hoeffding's Bound]

	Let $Y_{1},\cdots,Y_{m}$ independent random variables, $a_{1},\cdots,a_{m}$ and $b_{1},\cdots,b_{m}$ with $a_{s} < b_{s}$ such that:
	\[ \forall s, \mbox{ } \mathbb{P} \left( a_{s} \leq Y_{s} \leq b_{s} \right) = 1 \]

	We set $Z_{m} = \sum_{s=1}^{m} Y_{s}$, then:
	\[ \mathbb{P} \left( |Z_{m} - \mathbb{E}(Z_{m}) | \geq t \right) \leq 2\exp\left( - \frac{2t^{2}}{\sum_{s=1}^{m} (b_{s} - a_{s})^{2}} \right)  \]
		
	\end{proposition}

	In order to distinguish both hypotheses, we set $t = \frac{E_{0} - E_{1}}{2}$. So under $\mathcal{H}_{l}$, we have

	\begin{eqnarray*}
	\mathbb{P} \left( \left|V_{m} - E_l\right|   \geq  \frac{E_{0} - E_{1}}{2}  \right) & = & \mathbb{P} \left( \left|V_{m} - E_l\right|  \geq  \sum_{j=1}^{n} \frac{m_{j}}{2} (\varepsilon_{0}(j) - \varepsilon_{1}(j))^{2} \right)   \\
	 &\leq & 2 \exp \left( - \frac{2/4 \left( \sum_{j=1}^{n} m_{j} (\varepsilon_{0}(j) - \varepsilon_{1}(j))^{2} \right)^{2} }{\sum_{j=1}^{n} m_{j}  (\varepsilon_{0}(j) - \varepsilon_{1}(j))^{2}}  \right) \\
	& = &2\exp \left( -\frac{1}{2} \left( \sum_{j=1}^{n} m_{j} (\varepsilon_{0}(j) - \varepsilon_{1}(j))^{2} \right) \right) 
	\end{eqnarray*} 
	
%
%
	
	We decide that hypothesis $\mathcal{H}_{1}$ holds if $V_m < \frac{E_0+E_1}{2}$ and that $\mathcal{H}_{0}$ holds otherwise.
	It is clear that the probability $P_{e}$ to make a wrong decision with this distinguisher is smaller than $2 e^{- \frac{1}{2} \left( \sum_{j=1}^{n} m_{j} (\varepsilon_{0}(j) - \varepsilon_{1}(j))^{2} \right) }$.  If we want $P_{e} \leq 2e^{-\eta}$ for any fixed $\eta$, $m_{1},\cdots,m_{n}$ have to be such that:
	\begin{equation}
		\label{eq:probError}
			\frac{1}{2} \sum_{j=1}^{n} m_{j} (\varepsilon_{0}(j) - \varepsilon_{1}(j))^{2} \geq \eta \Rightarrow  \sum_{j=1}^{n} m_{j} (\varepsilon_{0}(j) - \varepsilon_{1}(j))^{2} \geq 2\eta
	\end{equation}
	Note that this is really the right order (up to some contant factor) for the amount of equations which is needed (the Hoeffding bound captures well up to constant factors 
the probability of the error of the distinguisher in this case) and using optimal Bayesian decision does not allow to change up to multiplicative factors 
the number of equations that are needed for a fixed relative error weight.
	Now assume that 
	\begin{ass}
 	 \label{ass:polyEqPar} 
 	 If we can compute $m$ parity-check equations of weight $w$ in time $T$, we are able to compute $n \cdot m$ parity-check equations of this weight in time $O(nT)$. 
 	\end{ass}
	
	This assumption holds for all ``reasonable'' randomized algorithms producing random parity-checks with uniform/quasi uniform 
	probability as long as $n \cdot m$ is at most some constant fraction (with a constant $<1$) of the total number of parity-check equations.
Now we set $j_{0}$ such that:
	
	\begin{equation}
	\label{eq:leadingWeight} 
	m_{j_{0}}(\varepsilon_{0}(j_0) - \varepsilon_{1}(j_0))^{2} = \mathop{\max}\limits_{1 \leq j \leq n} \{ m_{j}(\varepsilon_{0}(j) - \varepsilon_{1}(j))^{2} \}
	\end{equation}
	Clearly if we take now instead of the original $m$ parity-check equations just the $n \cdot m_{j_0}$ parity check equations of 
	weight $j_0$ the probability does of error does not get smaller than the bound $2e^{-\eta}$ that we had before since
		\[ n \cdot  m_{j_{0}}(\varepsilon_{0}(j_0) - \varepsilon_{1}(j_0))^{2} \geq  \sum_{j=1}^{n} m_{j} (\varepsilon_{0}(j) - \varepsilon_{1}(j))^{2} \Rightarrow 2 e^{-\frac{1}{2}n\cdot m_{j_{0}}  (\varepsilon_{0}(j) - \varepsilon_{1}(j))^{2}} \leq 2e^{-\eta}   \]

 	So, under Assumption \ref{ass:polyEqPar} if our distinguisher with several weights has enough parity-check equations available, we are able in polynomial time to compute $n \cdot m_{j_{0}}$ parity-check equations of weight $w_{j_{0}}$ where $j_{0}$ is chosen such that (\ref{eq:leadingWeight}) holds and with these parity-check equations the distinguisher of Subsection \ref{bias} can work too. The complexity of statistical decoding without the phase of computation of the parity-check equations is the number of parity-check equations that it is needed. 
So, under Assumption \ref{ass:polyEqPar}, its complexity with our first distinguisher will be for each codelength  $n$ the same up to a polynomial mutiplicative factor as the 
complexity with the second distinguisher. Moreover, under Assumption \ref{ass:polyEqPar} the complexity of the computation of the parity-check equations that is needed for both distinguishers is the same up to a polynomial factor. As the $\varepsilon_{1}(j) - \varepsilon_{0}(j)$ are exponentially small in $n$, in order to have a probability of success which tends to $1$, the $m_{j}$'s of both distinguisher have to be of order $\tilde{O} \left( \frac{1}{(\varepsilon_{0}(j)-\varepsilon_{1}(j))^{2}}\right)$. It leads to the conclusion that the asymptotic exponent of the statistical decoding is the same with considering 
some well chosen weight or several weights. We stress that this conclusion is about an asymptotic study of the complexity of statistical decoding. Indeed, in practice Algorithms \ref{alg:gauss} and \ref{alg:fusion} can output many parity-check equations of weight ''close'' to $Rn/2$ and  $r + (Rn-l)/2$. It will be counter-productive not to keep them and use them with the distinguisher we just described.

\section{A simple way of obtaining moderate weight parity-check equations}
\label{sec:naive}	
	As we are now able to give a formula for $\pwt$ we come back to the algorithm \\ \texttt{ParityCheckComputation}$_{w}$ in order to estimate $\pwtc$. 
There is an easy way of producing parity-check equations of moderate weight by Gaussian elimination.
This is given in Algorithm \ref{alg:gauss}  that provides a method for finding parity-check equations of weight $w =\frac{Rn}{2}$ of an $[n,Rn]$ random code. 
Gaussian elimination (\texttt{GElim}) of  an $Rn \times n$ matrix $G_{0}$ consists in finding $U$ ($Rn \times Rn$ and non-singular) such that:
	\[ UG_{0} = \lbrack I_{Rn} | G' \rbrack \]

	$L_{j}(G)$ denotes the $j-$th row of $G$ in Algorithm \ref{alg:gauss}.

		\begin{algorithm}[H]
			\caption{\texttt{ParityCheckComputation}$_{Rn/2}$ \label{alg:gauss}}
			\begin{algorithmic}[1]
				\State Input : $G \in \mathbb{F}_{2}^{Rn \times n}, i \in \mathbb{N}$ 
				\State Output : $\sS_i$  /*\textit{$P_{Rn/2}$ parity-check equations}*/
				\State $\sS_i \leftarrow \lbrack \mbox{ } \rbrack$
				\While{$|\sS_i| < P_{Rn/2}$}
				\State $P \leftarrow$ random $n \times n$ permutation matrix 
				\State $\lbrack G' | I_{Rn}  \rbrack \leftarrow$ \texttt{GElim}($GP$) and if it fails return to line 5
				\State $H \leftarrow$ $\lbrack I_{n(1-R)} | {G'}^T    \rbrack$ /*Parity matrix of the code*/
				\For{$j=1$ to $n(1-R)$}
				\If{$L_{j}(H)_{i} = 1$ and $w_{H}(L_{j}(H)) = Rn/2$}
				\State $\sS_i \leftarrow \sS_i \cup \{ L_{j}(H) P^T\}$
				\EndIf 
				\EndFor 
				\EndWhile
				\State \Return $\sS$
			\end{algorithmic}
		\end{algorithm}

	Algorithm \ref{alg:gauss} is a randomized algorithm. Randomness comes from the choice of the permutation $P$. 
It is straightforward to check that  this algorithm  returns $P_{Rn/2}$ parity-check equations of weight $Rn/2$ in time $\tilde{O}\left( P_{Rn/2} \right)$.

	Now we set $\tau = H^{-1}(1-R)$. This relative weight, which corresponds to the Gilbert-Varshamov bound, is usually used to measure the efficiency of decoding algorithms. 
Indeed it corresponds to the critical error weight below which we still have with probability $1-o(1)$ a unique solution to the decoding problem. It can be viewed as 
the weight for which the decoding problem is the hardest, since the larger the weight the more difficult the decoding 
problem seems to be (this holds at least for all known decoding algorithms of generic linear codes). 
As a consequence of Propositions 2 and 4, we have the following theorem:

	\begin{theorem}
		\label{theobias}
		[Naive Statistical Decoding's asymptotic complexity]$ $

		With the computation of parity-check equations of weight $Rn/2$ thanks to 	
		\\ \emph{\texttt{ParityCheckComputation}}$_{Rn/2}$, we have:	
		\[ \pwta{R/2}{\tau}  =  \pwtca{R/2}{\tau}\]
		where $\pwta{R/2}{\tau}$ is given by Theorem \ref{biasSDecoding}.
	\end{theorem}

	Exponents (as a function of $R$) of Prange's ISD and statistical decoding are given in Figure \ref{prstatdec}. As we see the difference is huge. This version of statistical decoding can not
 be considered as an improvement over ISDs. However, as $\omega \mapsto \pwta{\omega}{\tau}$ for $\tau$ fixed is an increasing function in $\omega$, we have to study the case $\omega < R/2$. It is the subject of the next section. 
We will give there an algorithm computing efficiently parity-check equations of smaller weight than $Rn/2$.
However we also prove there that no matter how efficiently we perform the pre-computation step, any version of statistical decoding is worse than Prange's ISD. 
	
	\begin{figure}
		\centering
		\includegraphics[scale = 0.6]{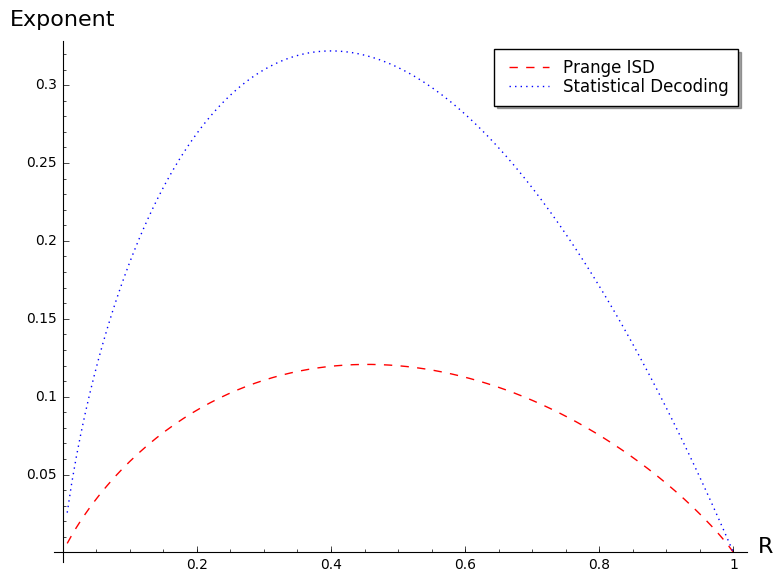}
		\caption{Asymptotic Exponents of Prange ISD and Statistical Decoding for $\tau=H^{-1}(1-R)$ et $\omega=R/2$} \label{prstatdec}
	\end{figure}

\section{Improvements and limitations of statistical decoding}
	\label{impvlim}

	\subsection{Framework}
		\label{fram}
	
	Before giving an improvement and giving lower bounds on the complexity of statistical decoding, we would like to come back to the computation problem of 
	the $\sS_i$'s in the complexity of statistical decoding. Our aim is to clarify
the picture a little bit. 	
	We stress that statistical decoding complexity is, if the $\sS_{i}$'s are already computed and stored, (up to a polynomial factor) the number of equations we use to take our decision. 
	We denote by $\Dc_w$ the part of statistical decoding which uses these parity-check equations to perform the decoding and by
$\Ac_w$ the  randomized algorithm used for outputting a certain number of random parity-check equations of weight $w$.
 \texttt{ParityCheckComputation}$_{w}$ is assumed to make a certain number of calls to $\Ac_w$. It is assumed that $\Ac_w$ outputs  $N_w$ parity-check equations of weight $w$ in time $T_w$ each time we run it.
	We assume that statistical decoding needs $\tilde{O}(P_w)$ equations. If we consider the computations of parity-check equations as part of statistical decoding, its complexity is given by:
	\[ \tilde{O} \left( P_w +  T_w \cdot \max (1,\frac{P_w}{N_w})  \right)  \]

	When $\frac{T_w}{N_w} = \tilde{O}(1)$, we say $\mathcal{A}_w$ gives equations in amortized time $\tilde{O}(1)$. With this condition if we assume $P_w \geq N_w$, the complexity is the number of equations needed.

	In any case, complexity of statistical decoding is lower-bounded by $\tilde{O}(P_w)$ and the lower the equation weight $w$, the 
	lower the number of equations $P_w$ we need for performing statistical decoding. 
	The goal of this section is to show how to find many parity-check equations of weight $<Rn/2$ in an efficient way and to give a minimal weight for which it makes sense to make this operation.

	\subsection{A lower bound on the complexity of statistical decoding}
		\label{lim}

	As we just 
	pointed out, statistical decoding needs $\tilde{O}\left(P_{w} \right)$ parity-check equations of weight $w$ to work.  Its complexity is therefore always greater than $\tilde{O}\left(P_{w}\right)$. We  assume 
again the code we want to decode to be a random code. This assumption is standard in the cryptographic context. 
The expected number of  parity-check equations of weight $w$ in an $[n,Rn]$ random binary linear code is $\frac{\binom{n}{w}}{2^{Rn}}$. 
Obviously if $w$ is too small there are not enough equations for statistical decoding to work, we namely need that
$$
P_w \leq  \frac{\binom{n}{w}}{2^{Rn}}.
$$
The minimum $\omega_{0}(R,\tau)$ such that this holds is clearly given by the minimal $\omega$ such that the following expression holds
\[ \pwta{\omega}{\tau} = H\left( \omega \right) - R \]

	So $\omega_{0}(R,\tau)$ gives the minimal relative weight such that asymptotically the number of parity-check equations needed for decoding is exactly the number of parity-check equations of  
 weight $w_{0}(R,\tau)$ in the code, where $w_0(R,\tau) \eqdef \omega_0(R,\tau)n$.
 Below this weight, statistical decoding can not work (at least not for random linear codes). 
 	In other words the asymptotic exponent of statistical decoding is always lower-bounded by $\pwta{w_{0}(R,\tau)}{\tau}$.

 	In the case of a relative error weight given by the Gilbert-Varshamov bound $\tau_{\text{DGV}}=H^{-1}(1-R)$, Theorem \ref{theobias} leads to the conclusion that
 	\[ \omega_{0}(R,\tau_{\text{DGV}}) =  \frac{1}{2} - \sqrt{\tau_{\text{DGV}} - \tau_{\text{DGV}}^2} \]

 	Moreover for all relative weights greater than $\omega_{0}(R,\tau_{\text{DGV}})$ the number of parity-check equations that are needed is exactly the number of parity-check equations of this weight that  exist in a random code. This result is rather intriguing and does not seem to have a simple interpretation. 
 The relative minimal weight $w_{0}(R,\tau_{\text{DGV}})$ is in relationship with the first linear programming bound of McEliece-Rodemich-Rumsey-Welch  and can be interpreted 
through its relationship with the zeros of Krawtchouk polynomials. This bound arises from the fact that from  Theorem \ref{theobias}, we know that $\omega_{0}(R,\tau_{\text{DGV}})$ corresponds to the relative 
weight where we switch from the complex case to the real case, and this happens precisely when we leave the region of zeros of 
the Krawtchouk polynomials.

	 Thanks to Figure \ref{fig:limita} which compares Prange's ISD, statistical decoding with parity-check equations of relative weight $R/2$ and $\omega_{0}(R,\tau)$ with $\tau = H^{-1}(1-R)$, we clearly see on the one hand that there is some room of improving upon naive statistical decoding based on parity-check equations of weight $Rn/2$, but on the other hand that even with the best improvement upon statistical decoding we might hope for, we will still be above the most naive information set decoding algorithm, namely Prange's algorithm.

	\begin{figure}
		\centering
		\includegraphics[scale = 0.6]{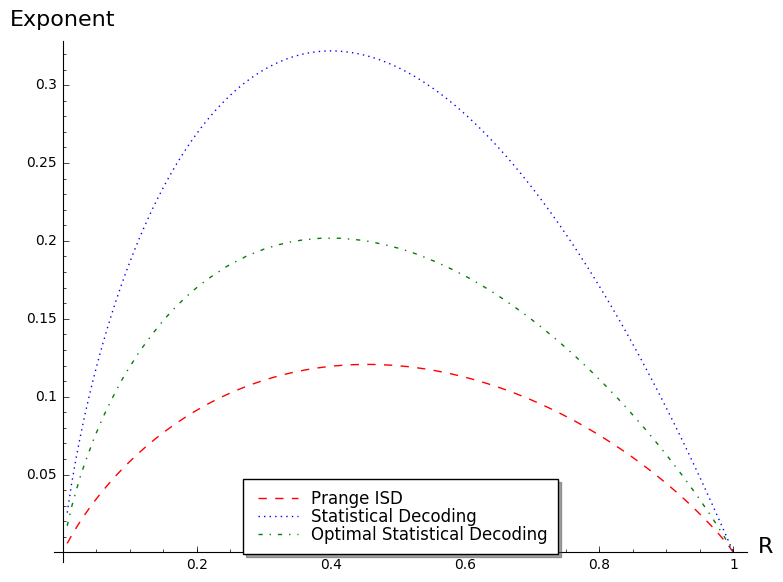}
		\caption{Asymptotic exponents of Prange ISD, naive statistical decoding and optimal/optimistic statistical decoding for $\tau=H^{-1}(1-R)$ \label{fig:limita}}
	\end{figure}

	\subsection{An improvement close to the lower bound}  
		\label{impv}	
	
	The goal of this subsection is to present an improvement to the computation of parity-check equations and to give its asymptotic complexity. R. Overbeck in \cite[Sec. 4]{O06} showed how to compute parity-check equations thanks to Stern's algorithm. We are going to use this algorithm too. However, whereas Overbeck used many iterations of this algorithm to produce a few parity-check equations of small weight, we observe that this algorithm produces in a natural 
way during its execution a large number of parity-check equations of relative weight smaller than $R/2$.
We will analyze this process here and show that it yields an algorithm $\mathcal{A}_w$ that gives equations in amortized time $\tilde{O}(1)$.

	To find parity-check equations, we described an algorithm which just performs Gaussian elimination  and 
	selection of sufficiently sparse rows.
	In fact, it is the main idea of Prange's algorithm. As we stressed in introduction, this algorithm has been improved rather significantly over the years (ISD family). Our idea to improve the search for parity-check equations is to use precisely these improvements. The first significant improvement is due to Stern and Dumer \cite{S88,D91}. The main idea is to solve a sub-problem with the birthday paradox. We are going to describe this process and show how it allows to improve upon naive statistical decoding.

	We begin by choosing a random permutation matrix $P \in \mathbb{F}_{2}^{n \times n}$ and putting the matrix $GP$ into the systematic form:
	
	\[
	\begin{bmatrix}
	I_{Rn-l} & G_{1} \\
	0 & G_{2}
	\end{bmatrix} \mbox{ where } G_{1} \in \mathbb{F}_{2}^{(Rn-l) \times (n(1-R)+l)} \mbox{ and } G_{2} \in \mathbb{F}_{2}^{l \times (n(1-R)+l)}
	\]
	\text{}
	
	1. We solve CSD($G_{2},r,0_{\lbrack l \rbrack}$).

	2. For each solution $e$, we output $e_{s} = (eG_{1}^T,e)P^T$.

	\begin{rem}
		We recall that solving CSD($G_{2},r,0_{\lbrack l \rbrack}$) means to find $r$ columns of $G_{2}$ which yield $0$. 
	\end{rem}

	$\cdot$ \textbf{Soundness:} We have 
$$Ge_s^T=GP\begin{bmatrix} G_1e^T \\ e^T \end{bmatrix} = \begin{bmatrix}
	I_{Rn-l} & G_{1} \\
	0 & G_{2}
	\end{bmatrix}\begin{bmatrix} G_1e^T \\ e^T \end{bmatrix} = \begin{bmatrix} G_1e^T + G_1e^T \\ G_2e^T \end{bmatrix}=0$$ and therefore $e$ is a parity-check equation of
	$\mathcal{C}$.

	$\cdot$ \textbf{Number of solutions:} The number of solutions is given by the number of solutions of 1. Furthermore, the complexity of this algorithm is up to a polynomial factor given by the complexity of 1.

	\begin{rem}
		This algorithm may not provide in one step enough solutions. In this case, we have to put $G$ in another systematic form ({\em i.e.} choose another permutation). The randomness of our 
algorithm  will come from this choice of permutation matrix.  
	\end{rem}

	$\cdot$ \textbf{Solutions' weight:} In our model $G$ is supposed to be random. So we can assume the same hypothesis for $G_{2}$. As the length of its rows is $Rn -l$, we get asymptotically parity-check equations of weight:
	\[ r+\frac{Rn-l}{2}(1 + o(1)) \]

	The first part of this algorithm can be viewed as the first part of ISD algorithms. There is a general presentation of these algorithms in \cite{FS09} in Section 3. All the efforts that have been spent to improve Prange's ISD can be applied to solve the first point of our algorithm. To solve this point, Dumer suggested to put $G_{2}$ in the following form:
	\[ G_{2} = \lbrack G_{2}^{(1)} | G_{2}^{(2)} \rbrack \mbox{ where } G_{2}^{(i)} \in \mathbb{F}_{2}^{l \times \frac{n(1-R) + l}{2}} \] 
	and to build the lists:
	\[ \mathcal{L}_{1} = \left\{ \left( e_{1},G_{2}^{(1)} e_1^T\right) \mbox{ } | \mbox{ } e_{1} \in \mathbb{F}_{2}^{\frac{n(1-R)+l}{2}} \mbox{ and } w_{H}(e_{1}) = \frac{r}{2} \right\} \]
	\[ \mathcal{L}_{2} = \left\{ \left( e_{2},G_{2}^{(2)}e_2^T \right) \mbox{ } | \mbox{ } e_{2} \in \mathbb{F}_{2}^{\frac{n(1-R)+l}{2}} \mbox{ and } w_{H}(e_{2}) = \frac{r}{2} \right\} \]

	Then we intersect these two lists with respect to the second coordinate and we keep the associated first coordinate. In other words, we get:
	\[ \{ (e_{1},e_{2}) \mbox{ } | \mbox{ } w_{H}(e_{i}) = \frac{r}{2} \mbox{ and } G_{2}^{(1)}e_1^T = G_{2}^{(2)}e_2^T \} \]

	\begin{rem}
		This process is called a fusion. 
	\end{rem}

	Algorithm \ref{alg:fusion}  summarizes this formally.

	\begin{algorithm}[H]
		\caption{\texttt{DumerFusion}\label{alg:fusion}}
			\begin{algorithmic}[1]
				\State Input : $G \in \mathbb{F}_{2}^{Rn \times n},l,r$.
				\State Output : $\sS$  /*\textit{subset of $\sH_{w}$}*/
				\State $\sS \leftarrow \lbrack \mbox{ } \rbrack$ /*\textit{Empty list}*/
				\State $\sT \leftarrow \lbrack \mbox{ } \rbrack$  /* \textit{Hash table}*/
				\State $P \leftarrow$ random $n \times n$ permutation matrix
				\State We find $U \in \mathbb{F}_{2}^{Rn\times Rn}$ non-singular such that $UGP = \begin{bmatrix}
					I_{Rn-l} & G_{1} \\
					0 & G_{2}
					\end{bmatrix}$ 
				\State We partition $G_{2}$ as  $\lbrack G_{2}^{(1)}|G_{2}^{(2)}\rbrack$ where $G_{2}^{(i)} \in \mathbb{F}_{2}^{l \times \left( \frac{n(1-R)+l}{2} \right)}$
				\ForAll{$e_{1} \in \mathbb{F}_{2}^{(n(1-R)+l)/2}$ of weight $r/2$}
				\State $x \leftarrow G_{2}^{(1)}e_1^T $
				\State $\sT\lbrack x \rbrack \leftarrow \sT \lbrack x \rbrack \cup \{ e_{1} \}$
				\EndFor 
				\ForAll{$e_{2} \in \mathbb{F}_{2}^{(n(1-R)+l)/2}$ of weight $r/2$}
				\State $x \leftarrow G_{2}^{(2)} e_2^T$
				\ForAll{$e_{1} \in \sT\lbrack x \rbrack$}
				\State $e \leftarrow (e_{1},e_{2})$
				\State $\sS\leftarrow \sS \cup \{ (eG_{1}^T,e)P^T \}$ 
				\EndFor 
				\EndFor 				
				\end{algorithmic}
			\end{algorithm}

		As we neglect polynomial factors, the complexity of Algorithm 3. is given by:
		\[ \tilde{O} \left( \binom{(n(1-R)+l)/2}{r/2} + \# \sS \right)   \]

		Indeed, we only have to enumerate the hash table construction (first factor) and the construction of $\sS$. In order to estimate $\#\sS$ we use the following classical proposition:

		\begin{prop}
			Let $L_{1},L_{2} \subseteq \{0,1\}^{l}$ be two lists where inputs are supposed to be random and distributed uniformly. Then, the expectation of 
the cardinality of their intersection is given by:
			\[ \frac{\# L_{1} \cdot \# L_{2}}{2^{l}} \]
		\end{prop}

		As we supposed $G_{2}$ random, we can apply this proposition to \texttt{DumerFusion}. Therefore,

		\begin{prop}[\texttt{DumerFusion}'s complexity]$ $

			\emph{\texttt{DumerFusion}}'s complexity is given by:
			\[ \tilde{O}\left( \binom{(n(1-R)+l)/2}{r/2} +  \frac{\binom{(n(1-R)+l)/2}{r/2}^{2}}{2^{l}} \right)   \]		
			and it provides on average
			\[ \frac{\binom{(n(1-R)+l)/2}{r/2}^{2}}{2^{l}} \]
			solutions
			
		\end{prop}

	In order to study this algorithm asymptotically, we introduce the following notations and relative parameters:
	\begin{nota}$ $

	$\cdot$ $N_{r,l} \eqdef \frac{\binom{(n(1-R)+l)/2}{r/2}^{2}}{2^{l}}$ ;

	$\cdot$ $T_{r,l} \eqdef   \binom{(n(1-R)+l)/2}{r/2} +  \frac{\binom{(n(1-R)+l)/2}{r/2}^{2}}{2^{l}}$ ;

	$\cdot$ $\rho = \frac{r}{n}$ ;

	$\cdot$ $\lambda = \frac{l}{n}$.
	\end{nota} 
We may observe that $N_{r,l}$ gives the number of parity-check equations that \texttt{DumerFusion} outputs in one iteration and $T_{r,l}$ 
is 
the running time of one iteration. 
There are many ways of choosing $r$ and $l$.
However
in any case (see Subsection \ref{lim}), as the weight of parity-check equations we get with \texttt{DumerFusion} is $(r + \frac{R - l}{2})(1+o(1))$ 
we have to choose $r$ and $l$ such that 		
$$w_{0}(R,t) \leq  r + (R-l)/2   $$		
which is equivalent to
		\begin{equation}
		\label{asym} 	
		\omega_{0}(R,\tau) \leq \rho + \frac{R- \lambda}{2} 
		\end{equation}

	 The following lemma gives 
	 an asymptotic choice of $\rho$ and $\lambda$ that allows
	 to get parity-check equations in amortized time $\tilde{O}(1)$:

	\begin{lemma}
		If 
		\begin{equation}
		\label{amtime}
		\rho = (1-R+\lambda) \cdot H^{-1}\left(\frac{2\lambda}{1-R+\lambda} \right)
		\end{equation} 
		\emph{\texttt{DumerFusion}} provides parity-check equations of relative weight $\rho + \frac{R-\lambda}{2}$ in amortized time $\tilde{O}(1)$. Moreover, with this constraint we have asymptotically :
		\[ N_{r,l} = \tilde{O} \left( 2^{\lambda \cdot n} \right)  \]

	\end{lemma}

	\begin{proof} We remark that $T_{r,l} = N_{r,l} + \binom{(n(1-R)+l)/2}{r/2}$. Our goal is to find $\rho, \lambda$ such that asymptotically $\frac{T_{r,l}}{N_{r,l}} = \tilde{O}(1)$. The constraint 
	\eqref{amtime} follows from $\binom{u}{v} = \tilde{O}\left( 2^{u \cdot H (u/v)} \right)$. 
	
	\end{proof}

	 We are now able to give the asymptotic complexity of statistical decoding with the use of \texttt{DumerFusion} strategy.

	\begin{theorem}  With the constraints (\ref{asym}), (\ref{amtime}) and 
	\begin{equation}
	\label{oneite}
	\lambda \leq \pwta{\rho + \frac{R-\lambda}{2}}{\tau}
	\end{equation} 	
	for $(\rho,\lambda)$ we have:
	\[  \pwtca{\rho + (R-\lambda)/2}{\tau}= \pwta{\rho + (R- \lambda)/2}{\tau}   \]	
	\end{theorem}

	\begin{proof} Thanks to (\ref{amtime}) and (\ref{oneite}) we use Subsection \ref{fram} and we conclude that under theses constraints we have $\pwta{\rho + (r-\lambda)/2}{\tau}=\pwtca{\rho + (r-\lambda)/2}{\tau}$. 
	\end{proof}

	\begin{rem} We 	summarize  the meaning of the constraints as:
	\begin{itemize}

	\item With (\ref{asym}) we are sure there exists enough parity-check equations for statistical decoding to work;

	\item With (\ref{amtime}) \texttt{DumerFusion} gives parity-check equations in amortized time $\tilde{O}(1)$;

	\item With (\ref{oneite}) \texttt{DumerFusion} 
	provides always  no more equations in one iteration 
	than we need.
	\end{itemize}

	\end{rem} 

	In order to get the optimal statistical decoding complexity we minimize $\pi(\rho + (R-\lambda)/2,\tau)$
	(with $\pi(\rho + (R-\lambda)/2,\tau)$ given by Theorem \ref{biasSDecoding})  under constraints \eqref{asym}, \eqref{amtime} and \eqref{oneite}. The exponent of statistical decoding with this strategy is given in Figure \ref{fig:limit}.
		
	\begin{figure}
		\centering
		\includegraphics[scale = 0.6]{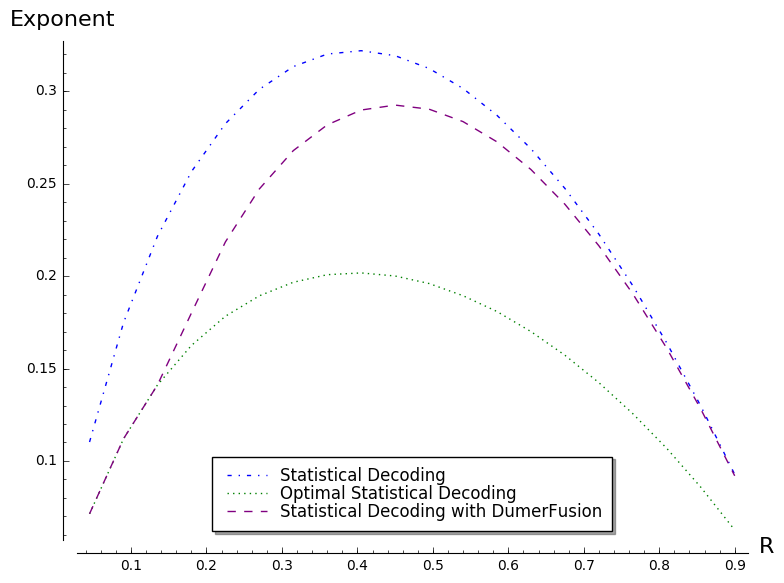}
		\caption{Asymptotic exponents of naive statistical decoding and with the use of optimal \texttt{DumerFusion} and optimal/optimistic statistical decoding for $\tau=H^{-1}(1-R)$ \label{fig:limit}}
	\end{figure}

	As we see, \texttt{DumerFusion} with our strategy allows statistical decoding to be optimal for rates close to $0$. 
	We can further improve \texttt{DumerFusion} with ideas of \cite{MMT11} and \cite{BJMM12}, however this comes at the expense of having a much more involved analysis and would not allow to go beyond the barrier of the lower bound on the complexity of statistical decoding given in the previous subsection. Nevertheless with the same strategy, these improvements lead to better rates with an optimal work of statistical decoding. 
	
\section{Conclusion}
	\label{concl}

In this article we have revisited statistical decoding with a rigorous study of its asymptotic complexity. We have shown that under Assumption 1 and 2 this algorithm is regardless of any strategy we choose for producing the moderate weight parity-check equations needed by this algorithm
always  worse than Prange ISD for the hardest instance of decoding (i.e. for a number of errors equal to Gilbert Varshamov bound).
In this case a very intriguing phenomenon happens, we namely need for a large range of parity-check weights all the parity-check available
in the code to be be able to decode with this technique. 
 It seems very hard to come up with choices of rate, error weight  and length for which
statistical decoding might be able to compete with ISD even if this can not be totally ruled out by the study we have made here.
However there are clearly more sophisticated techniques which could be used to improve upon statistical decoding. 
For instance using other strategies by grouping positions together and 
using all parity-check equations involving bits in this group could be another  possible interesting generalization of statistical decoding.

\end{document}